\documentclass[aps, physrev, superscriptaddress,
nofootinbib]
{revtex4-2}

\usepackage[utf8]{inputenc}
\usepackage{amsmath, amssymb,amsthm}
\usepackage{bm}
\usepackage{graphicx,color}
\usepackage[colorlinks=true,citecolor=blue,linkcolor=magenta]{hyperref}
\usepackage{tikz}
\usepackage{url}
\usepackage{tabularx}
\usepackage[italicdiff]{physics}
\usepackage{qcircuit}
\usepackage{algorithm}
\usepackage[noend]{algpseudocode}
\algrenewcommand\algorithmicdo{}

\usepackage[version=4]{mhchem}

\newcommand{\mr}[1]{\mathrm{#1}}

\newtheorem{prop}{Proposition}
\newtheorem{lemma}{Lemma}

\begin{document}

\title{Almost optimal measurement scheduling of molecular Hamiltonian \\ via finite projective plane}

\author{Wataru Inoue}
\affiliation{School of Engineering Science, Osaka University, 1-3 Machikaneyama, Toyonaka, Osaka 560-8531, Japan}
\affiliation{Center for Quantum Information and Quantum Biology, Osaka University, Japan}

\author{Koki Aoyama}
\affiliation{Graduate School of Engineering Science, Osaka University, 1-3 Machikaneyama, Toyonaka, Osaka 560-8531, Japan}
\affiliation{Graduate School of Information Science and Technology, Osaka University, 1-5 Yamada-oka, Suita, Osaka 565-0871, Japan}

\author{Yusuke Teranishi}
\affiliation{Graduate School of Information Science and Technology, Osaka University, 1-5 Yamada-oka, Suita, Osaka 565-0871, Japan}
\affiliation{Center for Quantum Information and Quantum Biology, Osaka University, Japan}

\author{Keita Kanno}
\author{Yuya O. Nakagawa}
\email{nakagawa@qunasys.com}
\affiliation{QunaSys Inc., Aqua Hakusan Building 9F, 1-13-7 Hakusan, Bunkyo, Tokyo 113-0001, Japan}

\author{Kosuke Mitarai}
\email{mitarai.kosuke.es@osaka-u.ac.jp}
\affiliation{Graduate School of Engineering Science, Osaka University, 1-3 Machikaneyama, Toyonaka, Osaka 560-8531, Japan}
\affiliation{Center for Quantum Information and Quantum Biology, Osaka University, Japan}

\date{\today}

\begin{abstract}
We propose an efficient and almost optimal scheme for measuring molecular Hamiltonians in quantum chemistry on quantum computers, which requires $2N^2$ distinct measurements in the leading order with $N$ being the number of molecular orbitals.
It achieves the state-of-the-art by improving a previous proposal by Bonet-Monroig \textit{et al.} [Phys. Rev. X \textbf{10}, 031064 (2020)] which exhibits $\frac{10}{3}N^2$ scaling in the leading order.
We develop a novel method based on a finite projective plane to construct sets of simultaneously-measurable operators contained in molecular Hamiltonians.
Each measurement only requires a depth-$O(N)$ circuit consisting of $O(N^2)$ one- and two-qubit gates under the Jordan-Wigner and parity mapping, assuming the linear connectivity of qubits on quantum hardwares.
We perform numerical simulation of our method for molecular Hamiltonians of hydrogen chains.
We count the number of sets of simultaneously-measurable operators generated by our method and estimate the number of measurement shots required to achieve the small standard deviation of the energy expectation value for precise quantum chemistry calculations.
Because evaluating expectation values of molecular Hamiltonians is one of the major bottlenecks in the applications of quantum devices to quantum chemistry, our finding is expected to accelerate such applications.
\end{abstract}

\maketitle

\section{Introduction}
One of the most promising, practical, and industry-relevant applications of quantum computers is quantum chemistry calculation~\cite{mcardle2018quantum, Cao2018}.
Measuring expectation values of operators (relevant to a physical or chemical system of interest) is an important and often indispensable subroutine in such application.
For example, the variational quantum eigensolver (VQE)~\cite{peruzzo2014variational, TILLY20221}, which has been intensively studied for utilizing noisy quantum computers called the noisy intermediate-scale quantum devices (NISQ devices)~\cite{Preskill2018}, is a method to prepare an approximate ground state of a given Hamiltonian by iteratively minimizing the expectation value of the Hamiltonian for a trial quantum state.
In VQE, measuring the expectation value of the Hamiltonian is essential and consists of the most important part of the algorithm.
Even for more advanced algorithms in the quantum chemistry applications of quantum computers, such as quantum phase estimation~\cite{kitaev1995, Cleve1998}, measuring expectation values of operators is important.
It is because the phase estimation algorithm allows us to sample eigenvalues and eigenstates of the Hamiltonian but it does not provide properties of the eigenstates.
For example, expectation values of operators with respect to the ground state are often required to investigate properties of simulated molecules or materials, such as the force acted on molecule~\cite{obrien2021efficient, Huggins2022} or transition amplitude with respect to some perturbation \cite{Ibe2022, 2210.00718v2}.
Considering that the operators arising in these applications share the same form with the Hamiltonian, we believe that a strategy for measuring the expectation value of the Hamiltonian will remain to be an essential tool even in the future when the quantum phase estimation is executable on fault-tolerant quantum computers.

The number of copies of quantum states needed to estimate  the expectation value of the Hamiltonian with sufficient accuracy becomes very large even for relatively small systems, due to the need to estimate expectation values of $O(N^4)$ distinct fermionic operators contained in the Hamiltonian for a system with $2N$ fermions (see Eq.~\eqref{eq: original mol Ham}).
Even for small systems corresponding to $\sim 20$ qubits, the required number of copies of quantum states can be over $10^9$ using naive strategies \cite{PhysRevResearch.4.033173}.
Gonthier \textit{et al.}~\cite{gonthier2020identifying} recently pointed out that several days may be needed to evaluate the expectation value of the Hamiltonian just one time to analyze the combustion energies of small organic molecules with sufficient accuracy, assuming the current speed of NISQ devices.
Improvement for the fast evaluation of the expectation values is therefore highly demanded for practical applications of quantum computers to quantum chemistry.
We note that the expectation value of the Hamiltonian can be evaluated by the fermionic one-particle and two-particle reduced density matrices (fermionic 1,2-RDMs, Eq.~\eqref{eq: def. of RDM}) and that such fermionic RDMs are utilized in several algorithms of quantum computational chemistry such as orbital-optimization~\cite{mizukami2020, Takeshita2020, Sokolov2020JCP}.

There have been various studies to tackle the problem of a large number of measurements for evaluating the expectation value of the Hamiltonian~\cite{mcclean2016theory,kandala2017,jena2019pauli,izmaylov2020,verteletskyi2020,yen2020,gokhale2020ON3,hamamura2020,zhao2020,crawford2021efficient, huggins2021efficient,yen2021cartan, BBOalgo, Huggins2022,yen2023deterministic,choi2022improving,Choi2023fluidfermionic}.
Especially, efforts have been put to reduce the number of the sets of the simultaneously-measurable operators in the Hamiltonian.
One can use commutativity or anti-commutativity of the Pauli operators in the qubit representation of the Hamiltonian to group those operators into simultaneously-measurable sets.
Bonet-Monroig, Babbush, and O'Brien~\cite{BBOalgo, private} showed that the number of such simultaneously-measurable sets to evaluate the fermionic 2-RDM, which is sufficient to determine the expectation value of the Hamiltonian, is lower bounded by $\Omega(N^2)$ for $2N$ fermion systems.
They also presented an algorithm to construct the $\frac{10}{3}N^2 + O(N)$ sets of simultaneously-measurable operators for evaluating the fermionic 2-RDM (see Sec.~\ref{subsec: discussion on comparison}), which is optimal up to a constant factor, by using a technique based on a binary partitioning of $2N$ integers.
Hereafter, we call their algorithm \textit{BBO algorithm} after the initials of the authors.

\textbf{Contribution.}
In this work, we provide an algorithm based on a finite projective plane in mathematics to construct $2N^2$ sets of simultaneously-measurable fermion operators which are sufficient to determine the expectation value of generic molecular Hamiltonian in quantum chemistry~[Eq.~\eqref{eq: original mol Ham}].
By using the symmetry of the coefficients of the quantum chemistry Hamiltonian~[Eq.~\eqref{eq: symm. of coef}], we classify terms in the Hamiltonian into several types and associate each of them with sets of simultaneously-measurable operators and quantum circuits to measure the operators in the set.
The crucial and novel contribution of our algorithm is the use of a finite projective plane to construct the sets of simultaneously-measurable operators for evaluating expectation values of the products of four distinct fermion operators (Sec.~\ref{subsubsec: clique for same spin}).
We formulate the problem of finding the sets of simultaneously-measurable operators as a minimal \textit{edge clique cover} problem of a certain graph, and find the almost optimal solution of it by using the finite projective plane.
The total number of the sets of simultaneously-measurable operators for the whole Hamiltonian scales as $2N^2 + O(N)$ for $2N$ fermion systems in our algorithm when $N=\Pi^K$ for a prime $\Pi$ and an integer $K>0$.
It improves the BBO algorithm which scales as $10N^2/3 + O(N)$ when applied in the same setup.
We also show that the quantum circuits for measuring those sets can be constructed by $O(N^2)$ one- and two-qubit gates with $O(N)$ depth in the case of Jordan-Wigner mapping~\cite{Jordan1928} and the parity mapping~\cite{BRAVYI2002210}.
We perform numerical simulation of our method for molecular Hamiltonians of hydrogen chains up to $N=30$ orbitals.
Our work gives a simple and efficient protocol to measure the expectation value of the molecular Hamiltonian in quantum chemistry so it can accelerate various applications of quantum computers to quantum chemistry calculation.
Moreover, since the minimum edge clique cover problem for general graphs is known to be NP-hard, our algorithm using the finite projective plane to create an almost optimal solution for the specific graph may be of an independent interest in the graph theory (we stress that our algorithm does not solve the minimum edge clique cover problem for \textit{general} graphs; it does solve for the specific graph corresponding to the quantum chemistry Hamiltonians).

This article is organized as follows.
We define our problem of measuring the expectation value of the Hamiltonian in Sec.~\ref{sec: problem def}.
Section~\ref{sec: main result} describes our main results: an algorithm to measure the expectation value of the Hamiltonian with $2N^2$ distinct quantum circuits.
We discuss our result and compare it with the previous studies in Sec.~\ref{sec: discussion}.
Summary and outlook are presented in Sec.~\ref{sec: summary}.
\section{Problem definition \label{sec: problem def}} 
In this section, we describe a concrete setup of the problem which we study. 

\subsection{Electronic structure Hamiltonian in quantum chemistry}
A central task in applying quantum computers to quantum chemistry calculations is to estimate an expectation value of the electronic structure Hamiltonian with $2N$ spin orbitals (fermions) in the form of
\begin{align}
\label{eq: original mol Ham}
    H = E_\mathrm{n} + \sum_{\sigma=\uparrow, \downarrow} \sum_{p,q=0}^{N-1}  h_{pq,\sigma}a^\dag_{p\sigma} a_{q\sigma} + \frac{1}{2} \sum_{\sigma,\tau=\uparrow, \downarrow} \sum_{p,q,r,s=0}^{N-1}  g_{pqrs,\sigma\tau}a^\dag_{p\sigma} a_{q\sigma} a^\dag_{r\tau} a_{s\tau},
\end{align}
where $E_\mathrm{n}$ is a nuclear repulsion energy (scalar), $a_{p\sigma}^\dag$ and $a_{q\sigma}$ are fermion creation and annihilation operators, respectively, satisfying the fermion anti-commutation relation $\{a_{p\sigma}^\dag, a_{q\tau}^\dag\} := a_{p\sigma}^\dag a_{q\tau}^\dag + a_{q\tau}^\dag a_{p\sigma}^\dag = 0, \{ a_{p\sigma}, a_{q\tau} \}= 0, \{a_{p\sigma}^\dag, a_{q\tau}\} = \delta_{pq} \delta_{\sigma\tau}$.
Two-electron integral $g_{pqrs, \sigma\tau}$ is defined as
\begin{align}
    g_{pqrs,\sigma\tau} = \iint \phi_{p\sigma}^*\left(\bm{r}_1\right)\phi_{q\sigma}\left(\bm{r}_1\right) \frac{1}{ |\bm{r}_{1} -\bm{r}_{2}|}  \phi_{r\tau}^*\left(\bm{r}_2\right)\phi_{s\tau}\left(\bm{r}_2\right) d \bm{r}_1 d \bm{r}_2,
\end{align}
and $h_{pq,\sigma}$ is defined as
\begin{align}
    h_{pq,\sigma} = - \frac{1}{2}\sum_{r} g_{prrq,\sigma\sigma} + \int \phi_{p\sigma}^*\left(\bm{r}_1\right)\left(-\frac{1}{2} \nabla_{\bm{r}_1}^2-\sum_{A=1}^{N_a} \frac{Z_A}{ |\bm{r}_{1} -\bm{R}_A| }\right) \phi_{q\sigma}\left(\bm{r}_1\right) d \bm{r}_1,
\end{align}
where $\phi_{p\sigma}(\bm{r})$ is a one-particle wavefunction for the orbital $p$ with spin $\sigma$, $N_a$ is the number of nuclei in the molecule, and $\bm{R}_A (Z_A)$ is the coordinate (charge) of the nucleus $A$, respectively.

When the orbital $\phi_{p\sigma}(\bm{r})$ is real-valued, which is the case for most calculations in quantum chemistry\footnote{Complex-valued orbitals are used, for example, when there is a magnetic field or the relativistic effect.}, there are symmetries in the coefficients of the Hamiltonian:
\begin{equation}
\label{eq: symm. of coef}
 g_{pqrs,\sigma\tau} = g_{qprs,\sigma\tau} = g_{pqsr,\sigma\tau}, \: h_{pq,\sigma}=h_{qp,\sigma}.
\end{equation}
Using these symmetries, we can rewrite the Hamiltonian as
\begin{align}
\label{eq: mol Hamiltonian}
 H = E_\mathrm{n} + \frac{1}{2} \sum_{\sigma=\uparrow, \downarrow} \sum_{p,q=0}^{N-1}  h_{pq,\sigma} A_{pq,\sigma} + \frac{1}{8} \sum_{\sigma,\tau=\uparrow, \downarrow} \sum_{p,q,r,s=0}^{N-1}  g_{pqrs,\sigma\tau} A_{pq,\sigma} A_{rs, \tau},
\end{align}
where $A_{pq, \sigma}$ is defined as
\begin{align}
    A_{pq, \sigma} := a^\dag_{p\sigma} a_{q\sigma}+a^\dag_{q\sigma} a_{p\sigma}.
\end{align}
We can therefore estimate the expectation value of Hamiltonian for a given state $\ket{\psi}$ by measuring $\ev{A_{pq,\sigma}} := \ev{A_{pq,\sigma}}{\psi}$ and $\ev{A_{pq,\sigma}A_{rs,\tau}} := \ev{A_{pq,\sigma}A_{rs,\tau}}{\psi}$ for all possible combinations of $p,q,r,s,\sigma,\tau$.
In the rest of this article, we describe how to measure such expectation values.

\subsection{Measurement cliques}
Mutually-commuting operators are simultaneously measurable.
We call a set of mutually-commuting operators \textit{a measurement clique}.
For a given measurement clique, we can associate it with one quantum circuit to perform projective measurement on simultaneous-eigenstates of all the operators in the measurement clique.
Repetitive execution of that circuit yields an estimate of an expectation value of any product of the operators in the measurement clique.
For example, when a measurement clique consists of two commuting operators $O_1$ and $O_2$ satisfying $[O_1,O_2]:= O_1O_2 - O_2O_1=0$, we can estimate the expectation values such as $\ev{O_1}, \ev{O_2}$, and $\ev{O_1 O_2}$ by the projective measurement on the simultaneous eigenstates of $O_1$ and $O_2$.

We want to construct a set of measurement cliques that covers all operators contained in the Hamiltonian~[Eq.~\eqref{eq: mol Hamiltonian}], namely,
\begin{equation}
 \mathbb{U} = \{ \mathcal{C}_1, \cdots, \mathcal{C}_M\}, 
 \quad
 \mathcal{C}_i = \left\{O_1^{(i)}, \cdots,O_{L_i}^{(i)}  \left| [O_j^{(i)}, O_k^{(i)}] = 0 \text{ for all } j,k\right. \right\},
\end{equation}
where $\mathcal{C}_i$ is $i$-th measurement clique, $M$ is the total number of the measurement cliques, $O_j^{(i)}$ is $j$-th operator in $\mathcal{C}_i$, and $L_i$ is the number operators in $\mathcal{C}_i$.
We require $\mathbb{U}$ to satisfy the condition that all $A_{pq,\sigma}$ and $A_{pq,\sigma} A_{rs,\tau}$ in the Hamiltonian [Eq.~\eqref{eq: mol Hamiltonian}] can be written as a single operator or a product of operators contained in some measurement clique $\mathcal{C}_i$.
If that condition is satisfied, we can measure all the expectation values of $\ev{A_{pq,\sigma}}$ and $\ev{A_{pq,\sigma} A_{rs,\tau}}$ by using $M$ distinct quantum circuits.
Therefore, the problem of estimating the expectation value of the Hamiltonian with a small number of distinct quantum circuits is reduced to find a set of measurement cliques whose number is as small as possible.
We construct such a set with $M \sim 2N^2 + O(N)$ in the next section.

\section{Main result: an efficient measurement  scheme for molecular Hamiltonians \label{sec: main result}}
We explain our main results in this section.
First we classify the terms in the molecular Hamiltonian~[Eq.~\eqref{eq: mol Hamiltonian}] into several types.
We then construct measurement cliques to cover all the types of the terms.
There are $O(N^4)$ terms in the Hamiltonian, but the number of measurement cliques can be reduced to $O(N^2)$ by the strategy that the terms in the Hamiltonian are reconstructed by a product of two operators in the clique.
Finally, we explain the quantum circuits associated with those measurement cliques and discuss their number of gates.

For the later use, we define the particle number operator of fermion as $n_{p\sigma} := a_{p\sigma}^\dag a_{p\sigma}$.
We also define $\bar{\sigma}$ representing the opposite spin of $\sigma$, i.e., $\bar{\uparrow} = \downarrow, \bar{\downarrow} = \uparrow$.

\subsection{Classification of terms}\label{sec:classification}
The Hamiltonian~\eqref{eq: mol Hamiltonian} consists of $A_{pq,\sigma}$ and $A_{pq,\sigma}A_{rs,\tau}$.
Independent terms coming from $A_{pq,\sigma}$ can be classified into two types by considering two cases $p=q$ and $p\neq q$:
\begin{itemize}
 \item[(1-1)] $n_{p\sigma}$ for $p=0,\cdots,N-1$ and $\sigma=\uparrow,\downarrow$.
 \item[(1-2)] $A_{pq,\sigma}$ for $0 \leq p < q \leq N-1$ and $\sigma=\uparrow,\downarrow$.
\end{itemize}
We note that we have only to consider $p>q$ when $p\neq q$ because $A_{pq, \sigma} = A_{qp,\sigma}$.

Similarly, independent terms coming from $A_{pq,\sigma}A_{rs,\tau}$ can be classified as follows.
When $\sigma \neq \tau$, there are three types of terms:
\begin{itemize} 
 \item[(2-1)] $n_{p\sigma} n_{q\bar{\sigma}}$ for $0 \leq p < q \leq N-1$ and $\sigma=\uparrow,\downarrow$,
  \item[(2-2)] $A_{pq,\sigma} n_{r\bar{\sigma}}$ for $0 \leq p < q \leq N-1, r=0,\cdots,N-1$, and $\sigma=\uparrow,\downarrow$,
  \item[(2-3)] $A_{pq,\sigma} A_{rs,\bar{\sigma}}$ for $0 \leq p < q \leq N-1, 0 \leq r < s \leq N-1$, and $\sigma=\uparrow,\downarrow$.
\end{itemize}
When $\sigma = \tau$, we have
\begin{itemize}
 \item[(2-4)] $n_{p\sigma}, n_{p\sigma} n_{q\sigma}$ for $0 \leq p < q \leq N-1$ and $\sigma=\uparrow,\downarrow$,
 \item[(2-5)] $A_{pq,\sigma}$ for $0 \leq p < q \leq N-1$ and $\sigma=\uparrow,\downarrow$ (this is the same as (1-2)),
 \item[(2-6)] $A_{pq,\sigma} n_{r\sigma}$ for $0 \leq p < q \leq N-1, r=0,\cdots,N-1, r\neq p, r\neq q$, and $\sigma=\uparrow,\downarrow$,
  \item[(2-7)] $A_{pq,\sigma} A_{rs,\sigma}$ with mutually different $p,q,r,s$ satisfying $0 \leq p < q \leq N-1, 0 \leq r < s \leq N-1$, and $\sigma=\uparrow,\downarrow$.
\end{itemize}
The number of terms for the cases (2-3)(2-7) is $O(N^4)$ and that for (2-2)(2-6) is $O(N^3)$ while the other cases (2-1)(2-4)(2-5) contain $O(N^2)$ terms.
Therefore, the challenge lies in how to cover the four cases (2-3)(2-7)(2-2)(2-6) with $O(N^2)$ measurement cliques.

\subsection{Construction of measurement cliques}
Here we construct a set of measurement cliques to cover all the types in the Hamiltonian discussed above.
The resulting measurement cliques are summarized as Tab. \ref{tab:cliques}.

\begin{table*}
    \caption{\label{tab:cliques} Summary of measurement cliques, the number of cliques for each type (when the number of molecular orbitals $N$ satisfies $N=\Pi^K+1$ for a prime $\Pi$ and an integer $K>0$), and corresponding terms in the Hamiltonian.}
    \begin{ruledtabular}
        \begin{tabular}{p{0.20\linewidth}p{0.40\linewidth}cp{0.20\linewidth}}
            Definition of cliques & Construction & Number of cliques & Corresponding terms described in Sec.~\ref{sec:classification}   \\ \hline
            $\mathcal{C}^{\mathrm{part}}$ (Eq.~\eqref{eq:Cpart def})& Particle number operators for all molecular orbitals and spins. (Sec.~\ref{sec:Cpart}) & 1 & (1-1), (2-1), (2-4) \\
            $\mathbb{U}^{[1]}$ (Eqs. \eqref{eq: mutually commuting int. set}-\eqref{eq:tildeC def}) & Iteration through all possible parings of integers $0,\cdots,N-1$, which can be done through scheduling of round-robin tournaments. (Sec.~\ref{sec:1RDM}) & $2(N-1)$ & (1-2), (2-2), (2-5)\\
            $\mathbb{U}^{[\mathrm{2, diff}]}$ (Eq.~\eqref{eq:U2diff def}) & Direct product of cliques $\mathcal{C}_{i,\sigma}^{[1]}$ [Eq.~\eqref{eq:def of C1}] for opposite spins. (Sec.~\ref{sec:2RDM-diff}) & $(N-1)^2$ & (2-3)  \\
            $\mathbb{U}^{[\mathrm{2, same}]}$ (Eq.~\eqref{eq: def of U2 same}) & Using a finite projective plane (Fig.~\ref{fig:projective-plane}) and this is our main contribution. (Sec.~\ref{subsubsec: clique for same spin}) & $(N-1)^2$ & (2-6), (2-7)
        \end{tabular}
    \end{ruledtabular}
\end{table*}

\subsubsection{Measurement clique for particle number operators \label{sec:Cpart}}
The first measurement clique is for the particle number operator:
\begin{equation}\label{eq:Cpart def}
 \mathcal{C}^\mathrm{part} = \{n_{0\uparrow}, \cdots,  n_{N-1\uparrow}, n_{0\downarrow}, \cdots,  n_{N-1\downarrow} \}.
\end{equation}
This is a measurement clique because all the particle number operators commute: $[n_{p\sigma}, n_{q\tau}]=0$.
$\mathcal{C}_\mathrm{part}$ can determine the expectation values of the terms of the types (1-1), (2-1), and (2-4).

\subsubsection{Measurement clique for $A_{pq,\sigma}$ \label{sec:1RDM}}
Next, we consider a measurement clique that can evaluate $\ev{A_{pq, \sigma}}$ with $0\leq p<q\leq N-1$.
By observing that $A_{pq, \sigma}$ and $A_{rs,\sigma}$ commute if $p,q,r,s$ are mutually different, a set of pairs of integers,
\begin{align}
\label{eq: mutually commuting int. set}
 I_i = \{ (p_1^{(i)}, q_1^{(i)}), \cdots, (p_{L_i}^{(i)}, q_{L_i}^{(i)}) \left| \text{$0 \leq p_j^{(i)} < q_j^{(i)} \leq N-1$ and $p_j^{(i)}, q_k^{(i)}$ for $j,k=1,\cdots,L_i $ are mutually different.} \right.\} 
\end{align}
can define a measurement clique
\begin{equation}\label{eq:def of C1}
 \mathcal{C}^{[1]}_{i,\sigma} = \{ A_{pq, \sigma} \left| (p,q) \in I_i \right. \}.
\end{equation}
Therefore, it suffices to find the sets $I_1, I_2, \cdots, I_{M^{[1]}}$ satisfying Eq.~\eqref{eq: mutually commuting int. set} and
\begin{equation}
\label{eq: completeness of pairs}
 \bigcup_{i=1}^{M^{[1]}} I_i = \{ (p,q) \left| 0 \leq p<q \leq N-1 \right. \}
\end{equation}
to estimate all expectation values of $\ev{A_{pq,\sigma}}$ with $p<q$.
For any given $p<q$, there is some $I_i$ such that $(p,q) \in I_i$ so that we can evaluate $\ev{A_{pq,\sigma}}$ by the measurement clique $\mathcal{C}^{[1]}_{i,\sigma}$. 
One of the ways to find such sets was proposed in the paper of BBO algorithm~\cite{BBOalgo} but we can also use various scheduling algorithm for round-robin tournaments.
The number of measurement cliques is $M^{[1]} = N -1$ ($M^{[1]} = N$) when $N$ is even (odd).

Finally, for measuring the terms of types (1-2), (2-2), and (2-5), we use the following set of measurement cliques.
\begin{align}
 \mathbb{U}^{[1]} &:= \mathbb{U}^{[1]}_{\uparrow} \cup \mathbb{U}^{[1]}_{\downarrow}, \\
 \mathbb{U}^{[1]}_{\sigma} &:= \{\tilde{\mathcal{C}}^{[1]}_{1,\sigma}, \cdots,  \tilde{\mathcal{C}}^{[1]}_{M^{[1]},\sigma} \},  \:
 \tilde{\mathcal{C}}^{[1]}_{i,\sigma} := \mathcal{C}^{[1]}_{i,\sigma} \cup \{n_{0\bar{\sigma}}, \cdots,  n_{N-1\bar{\sigma}} \}.\label{eq:tildeC def}
\end{align}
Note that we add the particle number operators of the opposite spin to $\mathcal{C}^{[1]}_{i,\sigma}$ to cover the type (2-2).
The total number of cliques contained in $\mathbb{U}^{[1]}$ is $2M^{[1]}$.

\subsubsection{Measurement clique for $A_{pq,\sigma}A_{rs,\tau}$ with different spin \label{sec:2RDM-diff}}

Expectation values $\ev{A_{pq,\sigma}A_{rs,\tau}}$ with different spin $\sigma\neq \tau$ can be measured with $(M^{[1]})^2$ distinct measurement cliques as follows.
We consider a direct product of the sets,
$\{\mathcal{C}^{[1]}_{1,\uparrow}, \cdots, \mathcal{C}^{[1]}_{M^{[1]},\uparrow} \} \times \{\mathcal{C}^{[1]}_{1,\downarrow}, \cdots, \mathcal{C}^{[1]}_{M^{[1]},\downarrow} \}$.
Namely, the set of measurement cliques
\begin{equation}
 \mathbb{U}^{[2,\mathrm{diff}]} := \left\{  \mathcal{C}^{[1]}_{i,\uparrow} \cup \mathcal{C}^{[1]}_{j,\downarrow} \left| i,j=1,\cdots, M^{[1]} \right. \right\} \label{eq:U2diff def}
\end{equation}
can evaluate all expectation values $\ev{A_{pq,\sigma}A_{rs,\tau}}$ with $\sigma \neq \tau$ (note that $[A_{pq,\sigma}, A_{rs,\tau}] = 0$ for any $p,q,r,s$ when $\sigma\neq\tau$).
The number of the measurement cliques in $\mathbb{U}^{[2,\mathrm{diff}]}$ is apparently $(M^{[1]})^2$.
The type (2-3) is covered by $\mathbb{U}^{[2,\mathrm{diff}]}$.

\subsubsection{Measurement clique for $A_{pq,\sigma}A_{rs,\tau}$ with the same spin
\label{subsubsec: clique for same spin}}
Expectation values $\ev{A_{pq,\sigma}A_{rs,\sigma}}$ and $\ev{n_{p\sigma}A_{rs,\sigma}}$ (the same spin case, the types (2-6) and (2-7)) are most non-trivial to construct measurement cliques. 
Here, we propose a novel algorithm to create the measurement cliques for them by using a projective finite plane.
This part is one of the main contributions of this study.

First, observe that the following equations
\begin{align}
 [A_{pq,\sigma}, A_{rs,\sigma}] = 0, \: [n_{p\sigma}, A_{rs,\sigma}] = 0,
\end{align}
hold if $p,q,r,s$ are mutually different.
We want to find sets of pairs of integers $J_1, J_2, \cdots, J_{M^{[2]}}$,
\begin{equation}
\label{eq: def of J_i's}
 J_i =  \{ (r_1^{(i)}, s_1^{(i)}), \cdots, (r_{M_i}^{(i)}, s_{M_i}^{(i)}) \left| 0 \leq r_j^{(i)} \leq s_j^{(i)} \leq N-1 \right.\},
\end{equation}
satisfying three conditions,
\begin{itemize}
 \item $r_j^{(i)}, s_k^{(i)}$ for $j,k=1,\cdots,M_i $ are mutually different except for the cases $r_j^{(i)}=s_j^{(i)}$.
 \item For any mutually distinct integers $0\leq p,q,r,s \leq N-1$ with $p<q$ and $r<s$, there exists some $J_i$ that contains $(p,q)$ and $(r,s)$.
 \item For any mutually distinct integers $0\leq p,r,s \leq N-1$ with $r<s$, there exists some $J_i$ that contains $(p,p)$ and $(r,s)$.
\end{itemize}
Note that the set $J_i$ allows a pair of the same integer like $(p,p)$ in contrast with $I_i$ in Eq.~\eqref{eq: mutually commuting int. set}.
Once such sets are found, the measurement cliques
\begin{equation}\label{eq: def of U2 same}
 \mathbb{U}^{[2, \mathrm{same}]} := \{ \mathcal{C}^{[2]}_{1}, \cdots, \mathcal{C}^{[2]}_{M^{[2]}} \}, \:
 \mathcal{C}^{[2]}_{i} = \mathcal{C}^{[2]}_{i,\uparrow}\cup\mathcal{C}^{[2]}_{i,\downarrow},\:
 \mathcal{C}^{[2]}_{i,\sigma} := \{ A_{pq, \sigma} \left| (p,q) \in J_i \right. \}.
\end{equation}
can evaluate all the expectation values $\ev{A_{pq,\sigma}A_{rs,\sigma}}$ and $\ev{n_{p\sigma}A_{rs,\sigma}}$ for any mutually distinct integers $0\leq p,q,r,s \leq N-1$ with $p<q$ and $r<s$.
We note that $A_{pp,\sigma} = 2n_{p\sigma}$.

\textbf{Formulation as a edge clique cover problem.}
We then focus on how to obtain the sets $J_i$ for a given $N$.
We formulate the problem of finding the sets $J_i$ as a clique cover of a graph.
Let us define a graph $G$ with a set of vertices
\begin{equation}
\label{eq: def of vertex of G}
 V = V_1 \cup V_2, \: V_1 =\{(p,p) |p=0,1,\cdots, N-1 \}, V_2 = \{ (p,q) |0 \leq p<q \leq N-1\}.    
\end{equation}
A vertex $v=(p,q)$ corresponds to an operator $A_{pq,\sigma}$.
Edges of the graph $G$ between two vertices $v_a = (v_{a,1}, v_{a,2})$ and $v_b = (v_{b,1}, v_{b,2})$ are placed if either of the following three cases is met: (1) $v_a, v_b \in V_1, v_{a,1}\neq v_{b,1}$, (2) $v_a \in V_1,  v_b \in V_2$ and $v_{a,1}, v_{b,1}, v_{b,2}$ are mutually different, and (3) $v_a, v_b \in V_2$ and $v_{a,1}, v_{a,2}, v_{b,1}, v_{b,2}$ are mutually different.
When an operator corresponding to a vertex $v=(p,q)$ is denoted $O_v$, i.e., $O_v:=A_{pq,\sigma}$, this graph $G$ has an edge between $v_a$ and $v_b$ if and only if $[O_{v_a}, O_{v_b}]=0$.
As an example, Fig.~\ref{fig:graphG} shows $G$ for $N=6$.
It means that a clique $C$ (a subset of vertices with edges between all pairs of vertices contained in the set) of $G$ can define a measurement clique.
Moreover, the edge between $v_a$ and $v_b$ can be associated with the operator $O_{v_a} O_{v_b}$, and all the operators $\ev{A_{pq,\sigma}A_{rs,\sigma}}$ and $\ev{n_{p\sigma}A_{rs,\sigma}}$ with any mutually distinct integers $0\leq p,q,r,s \leq N-1$ with $p<q$ and $r<s$ are associated to the edges of the graph with one-to-one correspondence.
Therefore, finding the sets of $J_i$ for measurement cliques is reduced to the problem of finding the \textit{edge} clique cover of the graph $G$.
Finding an edge clique cover of a given graph with the minimal number of the cliques is NP-hard in general.
However, in our case, we know the property of the graph well and we can explicitly construct the edge clique cover having an almost optimal number of cliques with the classical computational time of $O(N^3)$.
We note that our graph $G$ is different from those in the previous studies~\cite{verteletskyi2020,crawford2021efficient,zhao2020,yen2020,gokhale2020ON3,jena2019pauli,izmaylov2020,hamamura2020} where the Pauli terms of the Hamiltonian are vertices; rather, \textit{edges} correspond to the terms in the Hamiltonian in our graph similarly to the BBO's approach \cite{BBOalgo}.

\begin{figure}
    \centering
\includegraphics[width=0.5\linewidth]{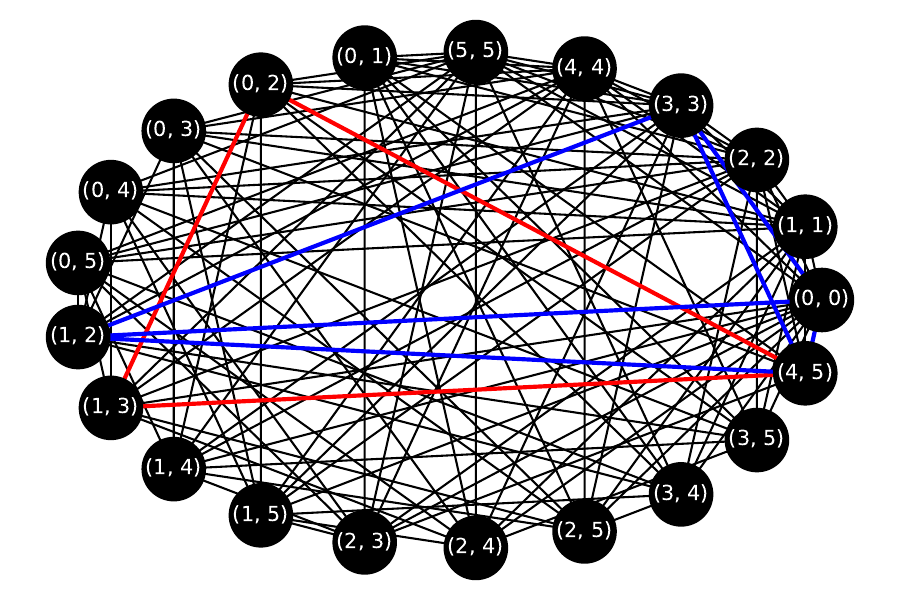}
    \caption{The graph $G$ defined in the main text for $N=6$. The bold colored edges correspond to the clique found by the procedure shown in Fig.~\ref{fig:projective-plane-example}. The red and blue edges corresponds to the cliques corresponding to $P_\gamma(4,3)$ and $P_\gamma(4,0)$, respectively.}
    \label{fig:graphG}
\end{figure}

\textbf{Construction of cliques using the finite projective plane.}
Our algorithm works when $N=\Pi^K+1$ for a prime $\Pi$ and an integer $K>0$.
When this is not the case, it suffices to choose the smallest $\Pi^K$ larger than $N-1$.
We consider a finite projective plane of the order $\Pi^K$ as the Desarguesian plane using the finite field $\mathbb{F}_{\Pi^K}$~\cite{albert2014introduction}.
Below, for simplicity, we restrict ourselves to the case of $K=1$ where the finite field is $\mathbb{Z}_\Pi$ and all the arithmetics (addition, multiplication, and their inverse operations) on the elements of the field are those of integers modulo $\Pi$. 
All the discussion described here applies to the case of arbitrary $K > 1$ by considering the corresponding arithmetics on the finite field.

The projective plane is defined by $\mathcal{P}$, a set called points, and $\mathcal{L}$, a family of the subset of $\mathcal{P}$ which is called lines.
The set $\mathcal{P}$ consists of three types of points: a single point $P_\alpha$, $\Pi$ points $\{P_\beta(y)\}_{y=0}^{\Pi-1}$, and $\Pi^2$ points  $\{P_\gamma(x,y)\}_{x,y=0}^{\Pi-1}$.
(see also Fig.~\ref{fig:projective-plane}).
The lines also consists of three types of the subset of points:  $\mathcal{L}=\{L_\alpha\}\cup\{L_\beta(i)\}_{i=0}^{\Pi-1}\cup\{L_\gamma(i,j)\}_{i,j=0}^{\Pi-1}$, defined as
\begin{align}
\begin{split}
    L_\alpha &= \left\{P_\alpha, P_\beta(0), \ldots, P_\beta(\Pi-1)\right\}, \\
    L_\beta(i) &= \left\{P_{\alpha}, P_{\gamma}(i, 0), \ldots, P_{\gamma}(i, \Pi-1)\right\}, \\
    L_\gamma(i,j) &= \left\{P_\beta(i), P_\gamma(k,(i k+j) \bmod \Pi)\right\}_{k=0}^{\Pi-1}.
\end{split}
\end{align}
The total number of lines is $\Pi^2+\Pi+1$.
Intuitively, the point $\{ P_\gamma(x,y) \}$ is interpreted as a two-dimensional plane with coordinates $x,y=0,\cdots, \Pi-1$, and the points $\{ P_\beta(x)\}$ and $P_\alpha$  are points at infinity.
$\mathcal{P}$ and $\mathcal{L}$ are known to constitute a finite projective plane.
This means that for any two points in $\mathcal{P}$ there exists only one line that contains (passes) both.
Also, for any two lines in $\mathcal{L}$, there exists only one point that is contained by both simultaneously.

Our idea to find a edge clique cover of $G$ begins by placing $\Pi+1 = N$ points on the plane.
Namely, we define $\mathcal{S}=\{ S(0), \cdots, S({\Pi}) \}$, a subset of $\mathcal{P}$, as (see Fig.~\ref{fig:projective-plane}) 
\begin{align}
 S(\Pi) = P_\alpha, \: S(k) = P_\gamma(k,k^2\bmod \Pi) \: (k=0,1,\cdots,\Pi-1).
\end{align}
We then associate a line $L_v$ of the plane to each vertex $v$ of the graph $G$ as follows.
For $v = (l, l) \in V_1$~(Eq.~\eqref{eq: def of vertex of G}), the line $L_v$ is defined as the line that passes $S(l)$ and does not pass other points in $\mathcal{S}$.
There always exists only one such line on the plane (see Appendix~\ref{appsubsec: exist-one-line} for proof).
For $v = (l, l') \in V_2$~(Eq.~\eqref{eq: def of vertex of G}), the line $L_v$ is defined as the line that passes both $S(l)$ and $S(l')$.
The property of the finite projective plane assures the existence and uniqueness of such line.
Moreover, distinct vertices are always associated to distinct lines, i.e., $v\neq v' \in V \Leftrightarrow L_v \neq L_{v'}$ holds because there is no line that passes three distinct points in $\mathcal{S}$ \cite{segre_1955,segrearxiv} (for completeness,  we prove this fact in Appendix \ref{appsubsec:no-three-point-proof}).

The complementary set of $\mathcal{S}$ in $\mathcal{P}$, $\mathcal{K} := \mathcal{P} \setminus \mathcal{S}$, consists of $(\Pi^2 + \Pi + 1) - (\Pi + 1) = \Pi^2$ points.
To a given point $p_k$ in $\mathcal{K}$ ($k=1,\cdots,\Pi^2$), we associate a set of vertices $C_k$ defined as (see Fig.~\ref{fig:projective-plane-example}),
\begin{equation}
\label{eq: def of c_k clique}
 C_k := \{ v \in V \: \text{such that $L_v$ passes $p_k$} \}.
\end{equation}
We claim that $\mathcal{C}_1, \cdots, \mathcal{C}_{\Pi^2}$ are cliques of the graph $G$ and constitute a edge clique cover of $G$.
\begin{prop}
 A set $\{ C_k \}_{k=1}^{\Pi^2}$ satisfies the followings:
 \begin{enumerate}
 \item[(a)] Each $C_k$ is a clique of the graph $G$.
 \item[(b)] $\{ C_k \}_{k=1}^{\Pi^2}$ is a edge clique cover of $G$. That is, for any edge between two distinct vertices $v$ and $v'$ on the graph $G$, there exists some $C_k$ satisfying $v, v' \in C_k$.
 \end{enumerate}
\end{prop}
\begin{proof}
To prove (a), consider two distinct vertices $v=(v_1, v_2), v'=(v'_1, v'_2)$ in $C_k$.
Two distinct lines $L_{v}, L_{v'}$ associated to these vertices intersects at $p_k \in \mathcal{K}$ by definition.
Suppose that there does not exist an edge between $v$ and $v'$.
If $v, v'\in V_1$, the non-existence of the edge implies $v_1 = v_2 = v'_1 = v'_2$, but this contradicts with $v\neq v'$.
If $v \in V_1$ and $v' \in V_2$, the non-existence of the edge implies $v_1 = v_2 = v'_1 \neq v'_2$ (or $v_1 = v_2 = v'_2 \neq v'_1$) holds and that $L_v$ $L_{v'}$ intersects at the point $S(v_1) \in \mathcal{S}$ other than $p_k$.
This contradicts with the uniqueness of the intersection between two lines.
Similarly, if $v,v' \in V_2$, the non-existence of the edge implies that $L_v$ and $L_{v'}$ intersects at some point in $\mathcal{S}$ other than $p_k$.
Again this contradicts with the uniqueness of the intersection between two lines.
Therefore, there must exist an edge between $v$ and $v'$ on the graph $G$, which proves (a).

To prove (b), let us take an edge between $v$ and $v'$ on the graph $G$ (note that $v\neq v'$).
There are two distinct lines $L_v$ and $L_{v'}$, and it has exactly one intersection $P \in \mathcal{P}$ on the plane.
The point $P$ must be in $\mathcal{K}$ due to the same argument to prove (a). 
Therefore, there is some $k$ satisfying $p_k = P$, and $C_k$ contains $v$ and $v'$.
\end{proof}

\begin{figure}
    \centering
    \includegraphics[width=0.45\linewidth]{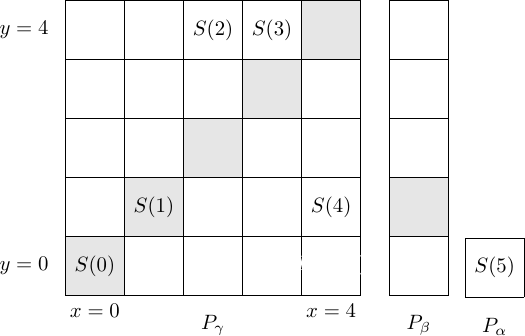}
    \caption{The projective plane used in our algorithm for $N=6 \: (\Pi=5)$. Each grid represents a point of the projective plane. As an example of a line, $L_\gamma(1,0)$ is represented by shaded grids. We label $N$ points $\{S(k)\}_{k=0}^{\Pi}$ such that a line passes either only one $S(k)$ or two $S(k)$'s. For each point that are not labeled $S(k)$, our algorithm lists all lines that pass it. The set of all lines that pass a specific point specifies a measurement clique, that is, the sets of operators that are simultaneously measurable. Examples are given in Fig.~\ref{fig:projective-plane-example}.
    }
    \label{fig:projective-plane}
\end{figure}

\begin{figure}
    \centering
    \includegraphics[width=\linewidth]{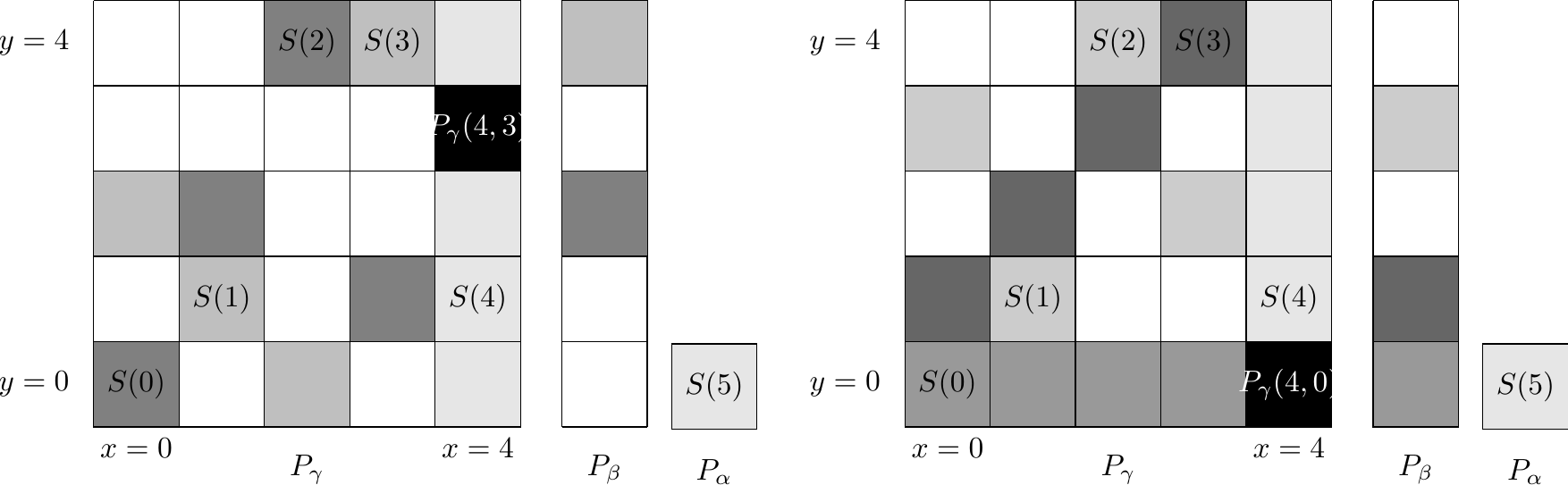}
    \caption{Example of how our algorithm works for $N=6 \: (\Pi=5)$. On the left, we show the lines that passes $P_\gamma(4,3)$. On the right, we show those for $P_\gamma(4,0)$. Grids with different colors correspond to distinct lines.
    (Left) 
    The darkest gray line passes $S(0)$ and $S(2)$, the next darkest does $S(1)$ and $S(3)$, and the lightest does $S(4)$ and $S(5)$. Therefore, the clique corresponding to $P_\gamma(4,3)$ is $C=\{ (0,2), (1,3), (4,5) \}$ which is colored red in Fig.~\ref{fig:graphG}. This clique corresponds to a set of operators $\{A_{02,\sigma}, A_{13,\sigma}, A_{45,\sigma}\}$ that are simultaneously measurable.
    (Right) 
    The darkest gray line passes $S(3)$ only.
    The next darkest does $S(0)$ only. 
    The next one does $S(1)$ and $S(2)$, and the lightest does $S(4)$ and $S(5)$. Therefore, the clique corresponding to $P_\gamma(4,3)$ is $C=\{ (0,0), (1,2), (3,3), (4,5) \}$ which is colored blue in Fig.~\ref{fig:graphG}. This clique corresponds to a set of operators $\{A_{00,\sigma}, A_{12,\sigma}, A_{33,\sigma}, A_{45,\sigma}\}$ that are simultaneously measurable.
    }
    \label{fig:projective-plane-example}
\end{figure}

\textbf{Definition of measurement clique.}
Now that we have a set of pairs of integers $\{ C_k \}_{k=1}^{\Pi^2} $ by using the finite projective plane, which satisfies Eq.~\eqref{eq: def of J_i's} and the conditions described there,
the measurement clique to evaluate all the terms $\ev{A_{pq,\sigma}A_{rs,\sigma}}$ and $\ev{n_{p\sigma}A_{rs,\sigma}}$ for both $\sigma=\uparrow, \downarrow$
can be explicitly constructed by taking $J_k = C_k$ in Eq.~\eqref{eq: def of U2 same}.
The number of the measurement cliques is $\Pi^2 = (N-1)^2$.

\subsubsection{Summary for measurement cliques}
To summarize, the measurement cliques we need to evaluate all the terms in the Hamiltonian [Eq.~\eqref{eq: mol Hamiltonian}] are
\begin{equation}
\label{eq: final cliques}
 \mathbb{U} = \{ \mathcal{C}^{\mathrm{part}}\} \cup \mathbb{U}^{[1]} \cup \mathbb{U}^{[2,\mathrm{diff}]} \cup \mathbb{U}^{[2,\mathrm{same}]}.
\end{equation}
The number of the total measurement cliques is $1 + 2(N-1) + (N-1)^2 + (N-1)^2 = 2N^2 -2N + 1 = 2N^2 + O(N)$.
The result is summarized in Tab. \ref{tab:cliques}.

Note again that, when $N=\Pi^K+1$ for a prime $\Pi$ does not hold, we can take the smallest $\Pi^K$ to construct $\mathbb{U}^{[2,\mathrm{same}]}$.
This strategy, in fact, upper-bounds the number of cliques contained in $\mathbb{U}^{[2,\mathrm{same}]}$ by $N^2+O(N(\log N)^2)$ for any practical $N$. 
It is because (c.f. Cram\'er's conjecture) $\Pi_{n+1}-\Pi_n < (\ln \Pi_n)^2$ holds for $7<\Pi_n\leq 5\times 10^{16}$, where $\Pi_n$ denotes $n$th prime \cite{nyman2003new}. 
This guarantees that, for a given (practical) $N$, we can choose a prime $\Pi>N$ such that $\Pi=N+O((\log N)^2)$ and run the algorithm to obtain $\Pi^2 = N^2+O(N(\log(N))^2)$ cliques.

\subsection{Quantum circuits for measuring operators in cliques}
After constructing the measurement cliques, we consider how to measure all the operators $A_{pq,\sigma}$ contained in each measurement clique simultaneously.
Strategies for measuring a single $A_{pq,\sigma}$ on a quantum computer depends strongly on the choice of the fermion-to-qubit mapping.
However, typical mappings (e.g., Jordan-Winger, parity, or Bravyi-Kitaev; see also~\cite{Steudtner_2018}) transform fermionic states with the fixed particle number, $\ket{f_{0\uparrow},\cdots,f_{N-1\uparrow},f_{0\downarrow},\cdots,f_{N-1\downarrow}} := \prod_{p,\sigma} (a_{p\sigma}^\dagger)^{f_{p\sigma}} \ket{\mathrm{vac}}$ with $f_{p\sigma}=0,1$ and $\ket{\mathrm{vac}}$ being the fermionic vaccum state, into a computational basis state of qubits and thus a particle number operator $n_{p\sigma}$ of fermions into a projection operator on computational basis states.
For such cases, measurement of $A_{pq,\sigma}$ reduces to a computational basis measurement after applying a unitary $U_{pq,\sigma}$ such that $U_{pq,\sigma}^\dagger A_{pq,\sigma}U_{pq,\sigma}$ can be written by the number operators.

In the language of fermions, $U_{pq,\sigma}$ can, for example, be
\begin{align}\label{eq:diagonalizing-unitary-example}
    U_{pq,\sigma} = \exp\left[-\frac{\pi}{4} (a^\dagger_{p\sigma}a_{q\sigma}-a^\dagger_{q\sigma}a_{p\sigma})\right]
\end{align}
which leads to
\begin{align}
    U_{pq,\sigma}^\dagger A_{pq,\sigma} U_{pq,\sigma} = n_{p\sigma} - n_{q\sigma}.
\end{align}
The challenge is how to efficiently implement such $U_{pq,\sigma}$ for all operators in a clique on a quantum computer with possibly limited connectivity.
One choice is to use the product of Eq.~ \eqref{eq:diagonalizing-unitary-example} for all $A_{pq,\sigma}$ in the clique.
The overall unitary in this case becomes a spin-conserving fermionic orbital rotation.
It is known that such a rotation can be performed using a depth-$(2N-3)$ quantum circuit on linearly connected qubits under the Jordan-Wigner encoding~\cite{Kivlichan2018}. 
This indeed is an efficient strategy, but below we will present even more efficient strategy which uses parallel application of nearest-neighbor fermionic swap gates for at most $N$ times.
Our construction is essentially parallel to that of BBO~\cite{BBOalgo}.

Suppose that we want to measure $m$ operators in a clique,
\begin{equation}
\label{eq: cliques to measure}
    \mathcal{C} = \{ A_{k_0 k_1,\sigma_0}, \cdots, A_{k_{2m-2} k_{2m-1}, \sigma_{m-1}} \},
\end{equation}
The measurement cliques other than $\mathcal{C}^{\mathrm{part}}$ can be written in this form.
The operators in $\mathcal{C}$ can always be reordered in up-then-down order as,
\begin{equation}
\label{eq: cliques to measure reordered}
    \mathcal{C} = \{ A_{k_0 k_1,\uparrow}, \cdots, A_{k_{2m_\uparrow-2} k_{2m_{\uparrow}-1},\uparrow}, A_{k_{2m_\uparrow} k_{2m_{\uparrow}+1},\downarrow},\cdots,A_{k_{2m-2} k_{2m-1}, \downarrow} \},
\end{equation}
where $m_{\sigma}$ is the number of $\sigma$-spin operators in the clique.
For simplicity, we explain the algorithm for the cases where $\mathcal{C}$ contains only $A_{kk',\sigma}$ with $k\neq k'$.
The extension of the algorithm to the cases with $A_{kk,\sigma}$ is straightforward.

Our strategy for measuring the operators in the clique $\mathcal{C}$ consists of two main steps.
\begin{enumerate}
 \item Apply nearest-neighbor fermionic swap gates~[Eq.~\eqref{eq: fswap}] to sort the indices of the operators in Eq.~\eqref{eq: cliques to measure reordered} into 
 \begin{equation}\label{eq: sorted cliques}
     \{A_{01,\uparrow},\cdots,A_{2m_\uparrow-2, 2m_{\uparrow}-1,\uparrow}, A_{01,\downarrow},\cdots,A_{2m_\downarrow-2, 2m_{\downarrow}-1,\downarrow}\}.
 \end{equation}
 The total number of the fermionic swap gates for this sorting is at most $N^2$ and the depth (counting a parallel fermionic swaps as one layer) is at most $N$.
 \item Transform the operators in Eq.~\eqref{eq: sorted cliques} to number operators.
 For example, this can be done by the parallel application of orbital rotations given in Eq. \eqref{eq:diagonalizing-unitary-example}.
However, as we will see later, it is not the only choice.
 By inspecting the concrete forms of $A_{j,j+1,\sigma}$ under specific fermion-qubit mappings, we can design simpler quantum circuits for diagonalizing them.
 We find that the Bell measurement circuit suffices our purpose for Jordan-Wigner mapping and a single layer of Hadamard gates does for parity and Bravyi-Kitaev mappings.
\end{enumerate}
We explain these two steps in order.

\subsubsection{Sorting operators in the clique by fermionic swaps}
We want to find a unitary $U^{\mathrm{swap}}$ such that
\begin{equation}
    (U^{\mathrm{swap}})^\dagger A_{k_ik_j,\sigma} U^{\mathrm{swap}} = A_{ij,\sigma},
\end{equation}
which is equivalent to the transformation from Eq.~\eqref{eq: cliques to measure reordered} to Eq.~\eqref{eq: sorted cliques}.
To this end, we utilize the fermionic swap gate,
\begin{equation}
\label{eq: fswap}
 f^{\mathrm{swap}}_{ij,\sigma} = 1 + a^\dag_{i\sigma} a_{j\sigma} + a^\dag_{j\sigma}a_{i\sigma} - a^\dag_{i\sigma} a_{i\sigma} - a^\dag_{j\sigma} a_{j\sigma}
\end{equation}
which is a Hermitian and unitary operator, i.e., $(f^{\mathrm{swap}}_{ij,\sigma})^2=I$. 
Its action can be described as
\begin{equation}
 f^{\mathrm{swap}}_{ij,\sigma} a_{i\sigma} f^{\mathrm{swap}}_{ij,\sigma} = a_{j\sigma}, \:
 f^{\mathrm{swap}}_{ij,\sigma} a_{j\sigma} f^{\mathrm{swap}}_{ij,\sigma} = a_{i\sigma}, \:
 f^{\mathrm{swap}}_{ij,\sigma} a_{k\tau} f^{\mathrm{swap}}_{ij,\sigma} = a_{k\sigma} \: (k \neq i,j~\lor~\sigma \neq \tau).
\end{equation}
We construct $U^{\mathrm{swap}}$ as a product of $M_{\uparrow}+M_{\downarrow}$ fermionic swap gates $U^{\mathrm{swap}} = U^{\mathrm{swap}}_\uparrow U^{\mathrm{swap}}_\downarrow$ where $U^{\mathrm{swap}}_\uparrow = f^{\mathrm{swap}}_{i_1j_1,\uparrow} \cdots f^{\mathrm{swap}}_{i_{M_\uparrow}j_{M_\uparrow},\uparrow}$ and $U^{\mathrm{swap}}_\downarrow = f^{\mathrm{swap}}_{i_1j_1,\downarrow} \cdots f^{\mathrm{swap}}_{i_{M_\downarrow}j_{M_\downarrow},\downarrow}$.
$U^{\mathrm{swap}}_\uparrow$ and $U^{\mathrm{swap}}_\downarrow$ can be constructed exactly in the same manner, so we consider how to construct $U^{\mathrm{swap}}_\uparrow$ below.

Define the following integers $p_l$ for $l=0,1,\cdots,N-1$ as
\begin{equation}
 p_l =
 \begin{cases}
   m & (\text{if there is $m\leq 2m_\uparrow-1$ such that $k_m=l$})\\
   - 1 & (\text{otherwise})
 \end{cases}.
\end{equation}
which indicates the position of the integer $l$ in the sequence $(k_0, k_1, ...,k_{2m_\uparrow-1})$. $p_l=-1$ corresponds to the case where $l$ is not found in the sequence.
Note that $p_l$ is well-defined and there is one-to-one correspondence between $(k_0, k_1, ...,k_{2m_\uparrow-1})$ and $(p_0, p_1, ...,p_{N-1})$ because we assume $k_0, k_1, \cdots, k_{2m_\uparrow-1}$ are mutually different. 
Applying fermionic swap gate $f_{ij,\uparrow}^\mathrm{swap}$ to the operators in the clique $\mathcal{C}$ invokes a change in $p_l$ as an exchange between $p_i$ and $p_j$ and corresponding change in $k$.
Therefore, the problem reduces to finding a sequence of swaps for the elements of $p_l$ which results in
\begin{equation}
p'_0 = 0, p'_1 = 1, \cdots, p'_{2m_\uparrow-1} = 2m_\uparrow-1,
p'_{2m_\uparrow} = \cdots = p'_{N-1} = -1,    
\end{equation}
because this $p'_l$ corresponds to
\begin{equation}
 k'_0 = 0, k'_1 = 1, \cdots, k'_{2m_\uparrow-1} = 2m_\uparrow-1.
\end{equation}

This problem can be solved by applying the odd-even sort algorithm~\cite{Habermann1972} to $p_l$.
This algorithm first compares the odd adjacent pairs $(p_1,p_2), \cdots, (p_{N-3}, p_{N-2})$ and exchanges two components in each pair if necessary.
Second, it does the same for the even adjacent pairs $(p_0,p_1), \cdots, (p_{N-2}, p_{N-1})$.
Each of these steps requires $N/2$ operations at most but they can be executed in parallel.
Iterating this process, it successfully sorts the integer with $N$ steps (i.e., with $N/2$ even and $N/2$ odd steps) at most.
Hence we can construct $U^{\mathrm{swap}}_\uparrow$ by interpreting exchanges in the sort algorithm as the fermionic swap gates.
$U^{\mathrm{swap}}_\downarrow$ can be constructed in the same manner and can be executed in parallel with $U^{\mathrm{swap}}_\uparrow$.
Therefore, $U^{\mathrm{swap}}$ in this construction consists of at most $N$ layers of at most $N$ parallel fermionic swap gates.

The actual representation of the femionic swap gate $f_{ij,\sigma}^\mathrm{swap}$ on qubits depends on the fermion-qubit mapping, but it becomes typically local when the swap is occurred at adjacent sites, $i=l, j=l+1$ for some integer $l$.
To see this, let us first define fermimonic operators $\tilde{a}_{j}$ with $j=0,1,\cdots,2N-1$ labeling both the orbital and the spin simultaneously by
\begin{equation}
 \tilde{a}_{j} :=
\begin{cases}
 a_{j, \uparrow} & (0 \leq j \leq N-1)\\ 
 a_{j-N, \downarrow} & (N \leq j \leq 2N-1)
\end{cases}.
\end{equation}
This is the so-called up-then-down convention for ordering spin-orbitals.
The Jordan-Wigner \cite{Jordan1928}, parity \cite{BRAVYI2002210}, and Bravyi-Kitaev \cite{BRAVYI2002210, Seeley2012, PhysRevA.95.032332} transformation respectively maps $\tilde{a}_{k}$ to
\begin{align}\label{eq:transforms}
    \tilde{a}_j \to \left\{\begin{array}{ll}
        (X_j+iY_j)Z_{j-1}\cdots Z_{0}/2 & \text{(Jordan-Wigner)} \\
        \left( Z_{j-1} X_j  + i Y_j \right) X_{j+1} \cdots X_{2N}/2 & \text{(parity)} \\
        \left( X_{U(j)} Z_{P(j)}  + i X_{U(j)\setminus\{j\}} Y_j Z_{P(j+1)\setminus \{j\}} \right)/2 & \text{(Bravyi-Kitaev)}
    \end{array}
    \right.,
\end{align}
where $U(j)$ and $P(j)$ are set of indices defined in Appendix \ref{appsec: meas. circ for PT-BK}, and the notation like $X_{U(j)}$ means $\prod_{l\in U(j)} X_{l}$.
From the definition of $f_{i,i+1,\sigma}^\mathrm{swap}$ (Eq.~\eqref{eq: fswap}), $f_{i,i+1,\sigma}^\mathrm{swap}$ becomes a two- and three-qubit operator in the Jordan-Wigner and parity mappings, respectively, because the chain of Pauli $Z$ operators and $X$ operators cancels out.
For the Bravyi-Kitaev mapping, it becomes a $O(\log N)$-qubit operator since the number of elements in the sets $U(j)$ and $P(j)$ is $O(\log N)$.

\subsubsection{Diagonalizing $A_{j,j+1,\sigma}$ in specific fermion-to-qubit mappings
\label{subsubsec: diagonalization circuit}}
Let $\tilde{A}_{k k'} := \tilde{a}_k^\dag \tilde{a}_{k'} + \tilde{a}_{k'}^\dag \tilde{a}_{k}$.
Note that the sorted clique (Eq.~\eqref{eq: sorted cliques}) only contains $\tilde{A}_{2j,2j+1}$ and therefore we only need to consider how to diagonalize them.
Under three different fermion-to-qubit mappings, $\tilde{A}_{2j,2j+1}$ becomes
\begin{equation}
\label{eq: JW hopping term}
 \tilde{A}_{2l,2l+1} \to \left\{
 \begin{array}{ll}
     (X_{2l} X_{2l+1} + Y_{2l} Y_{2l+1})/2 & \text{(Jordan-Wigner)} \\
     \left( - Z_{2l-1} X_{2l} Z_{2l+1} + X_{2l} \right)/2 & \text{(parity)} \\
     \left( - Z_{2l+1}X_{2l}Z_{P(2l+2)\setminus\{2l+1\}}Z_{P(2l)} + X_{2l} \right)/2 & \text{(Bravyi-Kitaev)}
 \end{array}
 \right..
\end{equation}
It is easy to see that, for the Jordan-Wigner mapping, we can diagonalize it by applying the Bell measurement circuit to qubits $2l$ and $2l+1$.
Also, it can be seen that the set $P(2l+2)\setminus\{2l+1\}$ and $P(2l)$ contains odd integers $i< 2l$ by inspecting Algorithm \ref{alg:P} in Appendix \ref{appsec: meas. circ for PT-BK}.
This means that, for the parity and Bravyi-Kitaev mappings, the Hadamard gate to all even-site qubits suffices our purpose.
The total number of local gates are $O(N)$ and the depth is $O(1)$ in all of three cases.

\section{Numerical simulation for hydrogen chains \label{sec: numerics}}
In this section, we apply our method to the molecular Hamiltonians for hydrogen chains.
We compute the number of measurement cliques for these Hamiltonians and estimate the number of measurement shots to realize the standard deviation of the energy typically required by precise quantum chemistry calculations.

\subsection{Setup}
We consider hydrogen chains $\ce{H4},\ce{H6},\ce{H8},\cdots,\ce{H30}$ where the atomic distance between two adjacent hydrogens is 1~\AA.
We employ the STO-3G basis set to perform the Hartree-Fock calculation with the numerical library PySCF~\cite{pyscf_1,pyscf_2}, and the Hartree-Fock orbitals are used to construct the molecular Hamiltonian [Eq.~\eqref{eq: original mol Ham}].
Note that the number of orbitals $N$ is the same as the number of hydrogens with this setup (e.g.,  $N=4$ for \ce{H4}).
The Jordan-Wigner transformation~\cite{Jordan1928} is used to map the fermionic Hamiltonian into the qubit one, implemented by the library OpenFermion~\cite{openfermion}.

We compare our method with two existing grouping techniques: qubit-wise commuting (QWC)~\cite{kandala2017, verteletskyi2020} and general commuting (GC) grouping~\cite{yen2020, crawford2021efficient}.
Both methods divide the terms in the Hamiltonian~\eqref{eq: original mol Ham} into simultaneously-measurable Pauli operators.
Note that the methods based on the factorization of fermionic Hamiltonians is more efficient than these \cite{huggins2021efficient}.
However, since such methods require us to apply fermionic orbital rotation operations which are non-Clifford and need more gates depending on hardware, they cannot be directly compared to our methods.
We thus omit comparison with them in this work.
For QWC grouping, the quantum circuits for measurement consist of $O(N)$ one-qubit Clifford gates with depth one.
For GC grouping, the number of two-qubit gates (CZ or CNOT, both are Clifford) for the measurement circuits is $O(N^2 / \log N)$~\cite{crawford2021efficient, Aaronson2004}.
We implemented the ``sorted insertion'' algorithm for both groupings~\cite{crawford2021efficient}.
The details of both grouping methods and their implementations can be found in Appendix~\ref{appsec:grouping-details}.

After the grouping, we estimate the number of shots required to suppress the standard deviation of the energy expectation value down to $10^{-3}$ Hartree.
This value ($10^{-3}$ Hartree) is comparable to the so-called chemical accuracy that is typically targeted at in precise quantum chemistry calculations.
The concrete procedures to estimate the number of shots are in Appendix \ref{appsec: number of shots}.

\subsection{Numerical results}
We performed grouping of the hydrogen chain Hamiltonians by our method up to $N=30$, while we did it by QWC and GC grouping up to $N=20$.
This is because the classical computational time is much longer for the two (QWC and GC) methods as illustrated in Fig.~\ref{fig:time}.
Our implementations of grouping methods are pure-Python strongly depending on OpenFermion \cite{openfermion}, works as a single-thread, and is not optimized for speed.
The CPU used in this experiment is AMD EPYC 7252.
Figure~\ref{fig:time} shows that our method is highly efficient compared to the existing methods.

Numerical results are shown in Fig.~\ref{fig:numerics}.
In Fig.~\ref{fig:numerics}(a), we compare the number of groups (measurement cliques) obtained by each grouping method.
To see the scaling of the number of groups, we also fit the results by the function $a N^{b}$, where $a,b$ are the fitting parameters, by the linear regression of the data points with $N\geq 10$ on the log-log plot. 
This fitting yields $b= 2.03(7)$ for our method, $b= 2.34(3)$ for GC grouping, and $b= 4.017(3)$ for QWC grouping.
We observed that GC grouping exhibits the slightly smaller number of groups than that of our method for $N \leq 20$.
However, the predicted scaling exponent of GC grouping is larger than that of our method. 
We expect that our method beats GC grouping for larger $N$ in terms of the number of groups (see the fitting lines in the figure).
We also note that the number of groups of our method seems to follow the theoretical scaling $2N^2 - 2N + 1$ (the total number of all cliques~[Eq.~\eqref{eq: final cliques}]).


In Fig.~\ref{fig:numerics}(b), we compare the number of shots required to make the standard deviation of the energy $10^{-3}$ Hartree.
We again performed the fitting of the data by the function $a N^{b}$ in the same way as we did for the number of groups, which results in $b= 3.51(5)$ for our method, $b= 3.12(7)$ for GC grouping, and $b= 4.542(7)$ for QWC grouping.
We obeserve that GC grouping shows the smaller number of shots than that of our method.
In contrast to the number of groups, the scaling exponent is smaller in GC grouping, indicating that it would performed well than our method if we could execute GC grouping for larger $N > 20$.
This is because ``sorted insertion'' algorithm performs grouping considering the magnitudes of coefficients of Pauli operators while our method are not aware of such information. 
We stress that, however, the huge classical computational cost will prevent GC grouping to be applied for larger $N$.
Our method can be seen as a practical and effective way to group the Pauli operators in the Hamiltonian and reduce the number of shots for the desired accuracy of the energy expectation value.
It might be possible to further reduce the number of shots of our method by taking the information about coefficients into account.

\begin{figure}
    \centering
    \includegraphics[width=0.9\linewidth]{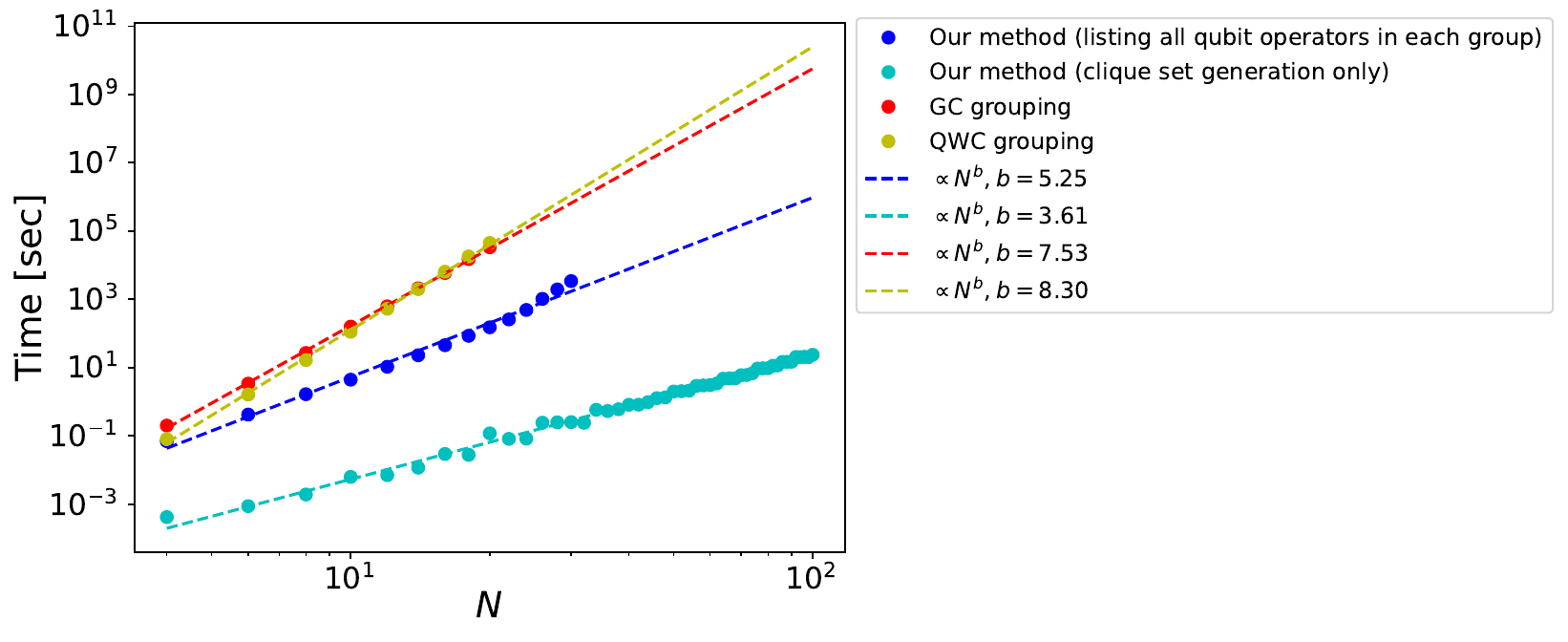}
    \caption{\label{fig:time} Computational time required for performing each grouping methods. We show two plots for our method. The blue one shows the time required for generating grouped qubit operators, that is, generating $\{G_i\}$ with $G_i = \{P^{i}_{j}\}$ where $P^{i}_j$ is a $2N$-qubit Pauli operators such that $[P^i_j, P^i_{j'}]=0$ for all $i,j,j'$. This is for a fair comparison with GC and QWC methods which generate such lists of groups. The light blue plot illustrates the time required to construct all cliques, represented by $\mathbb{U}$, defined as lists of integer pairs.}
\end{figure}

\begin{figure*}
 \includegraphics[width=\linewidth]{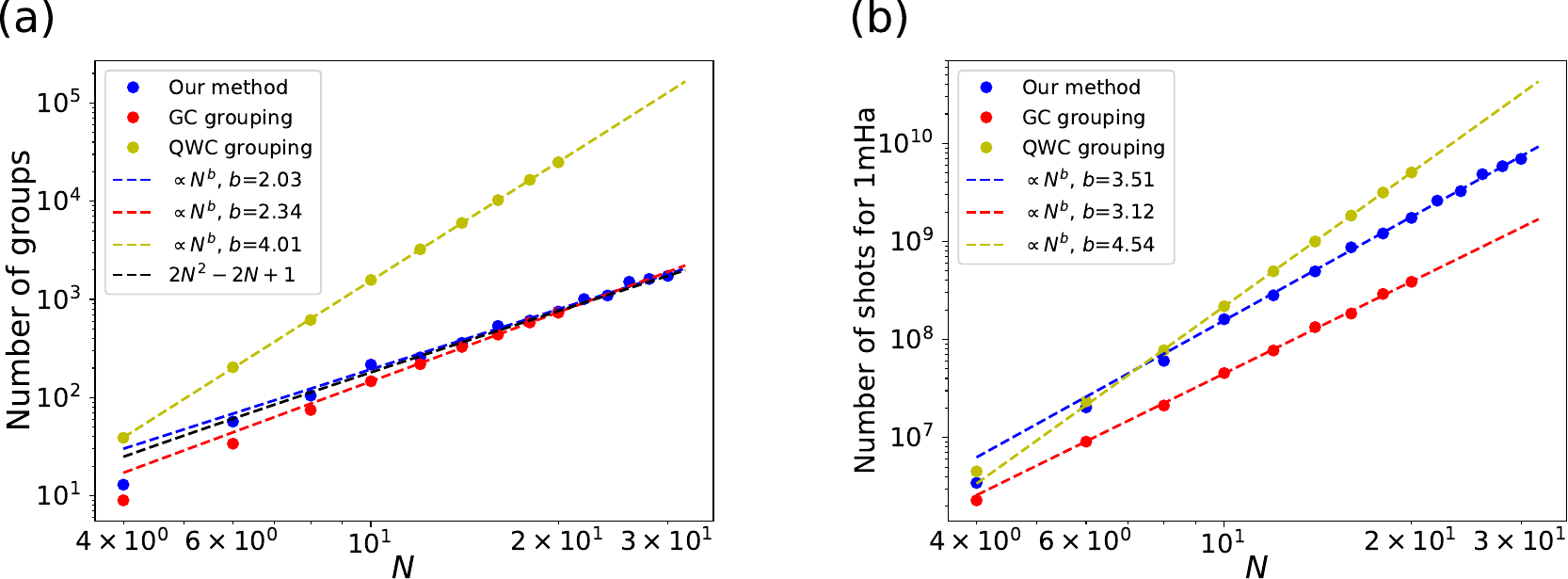}
 \caption{Numerical results for the molecular Hamiltonians for hydrogen chains $\ce{H}_N$.
 (a) The number of groups for each grouping method. Fitting lines for the data points with $N\geq 10$ are also shown.
 The Black dotted line denotes the theoretical scaling for our method, $2N^2-2N+1$.
 (b) The number of shots required to achieve the standard deviation of the energy lower than $10^{-3}$ Hartree, estimated by Haar random averaging (see main text).
 \label{fig:numerics}
 }
\end{figure*}

\section{Discussion \label{sec: discussion}}
We here discuss several points of our study and compare our result with the previous studies.

\subsection{Optimality of the obtained solution \label{sec:optimality}}
Our construction of cliques for the graph $G$ in Sec.~\ref{subsubsec: clique for same spin} is almost optimal from the viewpoint of the minimal edge clique cover problem.
To show this, we give a lower bound to the number of cliques to cover edges of $G$.
Consider a subgraph $G'$ where we remove all of vertices in $V_1$ and corresponding edges in $G$.
The number of cliques to cover all edges in $G'$ is clearly less than or equal to that of $G$, and it follows that a lower bound for the number of cliques for the edge cover of $G'$ is also a lower bound for $G$.
First, we claim that a clique in $G'$ can have at most $N/2$ vertices and thus $\frac{N}{2}(\frac{N}{2}-1)/2$ edges.
This is because there is an edge between two vertices $(p,q)$ and $(r,s)$ in $V_2$ if and only if $p,q,r,s$ are mutually different, and thus a clique corresponds to a pairing of $N$ integers.
Second, the number of edges in $G'$ is $N(N-1)(N-2)(N-3)/8$.
To see this, for every tuple of four integers $(p,q,r,s)$ such that $p<q<r<s$, there are three edges each of which connects the vertices $\{(p, q), (r, s)\}$, $\{(p, r), (q, s)\}$, and $\{(p, s), (q, r)\}$ in a graph $G'$.
The number of combinations of such integers is $\binom{N}{4}$, and therefore the statement follows.
Combining these two facts, we conclude that the number of cliques to cover all edges in $G'$ must be at least $[N(N-1)(N-2)(N-3)/8]/[\frac{N}{2}(\frac{N}{2}-1)/2] = (N-1)(N-3)$.
Therefore, our algorithm, where the number of cliques is $(N-1)^2$, provides an optimal solution to the minimal edge clique cover problem of $G$ up to the subleading order terms.

\subsection{Classical computational cost for constructing measurement cliques}
The classical computational cost for constructing the measurement cliques in our method is $O(N^3)$, which comes from creating a set $C_k$ in Eq.~\eqref{eq: def of c_k clique}.
For each $p_k \in \mathcal{K}$ (the number of $k$ is $\Pi^2=O(N^2)$), we add a vertex $v$ on $G$ to $\mathcal{C}_k$ by inspecting the line passing $p_k$ and $S(i)$ (the total number of $i$ is $\Pi+1 = N$): when the line passes $S(i)$ and $S(j)$, we add $(i,j)$ or $(j,i)$ to $\mathcal{C}_k$, and when the line passes only $S(i)$, we add $(i,i)$ to $\mathcal{C}_k$.
The exponent observed in Fig.~\ref{fig:time} is slightly worse than this theoretical scaling, which we attribute to additional logarithmic costs (the overhead happening when $N\neq \Pi+1$, memory access, integer arithmetics, etc.).
Note that this computational cost is optimal when we seek to find $O(N^2)$ measurement cliques, because the output of the algorithm must be of size $O(N^3)$ ($O(N^2)$ cliques with $O(N)$ vertices) by the discussion in Sec.~\ref{sec:optimality}.
BBO's paper \cite{BBOalgo} does not explicitly give the classical computational cost, but the size of the output of their algorithm is also $O(N^3)$ and thus its classical computational cost cannot be less than $O(N^3)$.

This $O(N^3)$ scaling of the classical computational cost is much smaller than the methods to create the measurement clique consisting of the terms of the Hamiltonian \textit{as is}~\cite{mcclean2016theory,kandala2017,jena2019pauli,izmaylov2020,verteletskyi2020,yen2020,gokhale2020ON3,hamamura2020,zhao2020}.
In those methods, one should calculate the commutativity or anti-commutativity of pairs of the $O(N^4)$ terms in the Hamiltonian, which results in as huge as $O(N^8)$ classical computational cost.
The target problem of the application of quantum computers in quantum chemistry lies for more than $N \sim 50$, and such scaling can be problematic in practice.

\subsection{Comparison of the number of measurement cliques in our algorithm with that of previous studies
\label{subsec: discussion on comparison}}
Our algorithm can determine the expectation value of the quantum chemistry Hamiltonian with $O(N^2)$ measurement cliques.
Again, this is advantageous over the conventional methods to create the measurement cliques consisting of the terms of the Hamiltonian~\cite{mcclean2016theory,kandala2017,jena2019pauli,izmaylov2020,verteletskyi2020,yen2020,gokhale2020ON3,hamamura2020,zhao2020}, in which
the total number of the measurement cliques is $O(N^3)$~\cite{gokhale2020ON3, zhao2020}.

 Let us compare our algorithm with BBO algorithm~\cite{BBOalgo}.
 The fermionic 1,2-RDMs are defined as
 \begin{equation}
 \label{eq: def. of RDM}
 \rho^{[1]}_{pq} = \ev{\tilde{a}_p^\dag \tilde{a}_q}, \rho^{[2]}_{pqrs} = \ev{\tilde{a}_p^\dag \tilde{a}_q^\dag \tilde{a}_r \tilde{a}_s}.
 \end{equation}
 for mutually distinct integers $p,q,r,s=0,\cdots,2N-1$.
 BBO algorithm can determine the fermionic 2-RDM with $O(N^2)$ measurement clique, so the scaling is the same as ours.
 It is based on the Majorana representation of the fermions and constructs measurement cliques of the Majorana operators by considering sets of integers similar to Eq.~\eqref{eq: def of J_i's}.
 That is, the products of two Majorana operators, rather than the products of four Majorana operators, consist the measurement cliques to evaluate the fermionic 2-RDM (represented as products of four Majorana operators). 
 The authors (BBO) \cite{BBOalgo, private} claimed that the number of measurement cliques to evaluate the fermionic 2-RDM is
\begin{align}
    4N^2 \left(\frac{16}{3}4^{-N_\mathrm{sym}} + 2^{1-N_\mathrm{sym}}\right),
\end{align}
where $N_\mathrm{sym}$ is the number of Pauli operators $W$ such that $[H,W]=0$\footnote{The expression differs by the factor of $4$ from the BBO's original paper \cite{BBOalgo}. It comes from the different definition of $N$, that is, we define $N$ as the number of molecular orbitals while BBO defines $N$ as the number of spin-orbitals.}.
For general Hamiltonians for quantum chemistry calculation [Eq.~\eqref{eq: original mol Ham}], typically we have two such symmetries: the numbers of spin-up and spin-down electrons~\cite{Bravyi2017}.
Putting $N_\mr{sym}=2$ in the above equation yields
\begin{equation}
     4N^2 \left(\frac{16}{3}4^{-2} + 2^{-1}\right) = \frac{10}{3}N^2,
\end{equation}
which is slightly larger than our method, $2N^2$.
We note that our method determines only the ``symmetric part" of the fermionic 2-RDM,
\begin{equation}
\ev{\Tilde{A}_{pq}\Tilde{A}_{rs}}
 = \rho^{[2]}_{pqrs} + \rho^{[2]}_{qprs} + \rho^{[2]}_{pqsr} + \rho^{[2]}_{qpsr},
\end{equation}
which is sufficient to determine the expectation value of the molecular Hamiltonian [Eq.~\eqref{eq: original mol Ham}] with the symmetries~\eqref{eq: symm. of coef}.
In contrast, BBO algorithm evaluates all elements of the fermionic 2-RDM.

\section{Summary and outlook \label{sec: summary}}
In this study, we have proposed a measurement scheme for general molecular Hamiltonians in quantum chemistry~[Eq.~\eqref{eq: original mol Ham}].
We have used the symmetries of the coefficients in the Hamiltonian and reduced the problem of evaluating the expectation value of the Hamiltonian into the evaluation of the terms $\ev{A_{pq,\sigma}}$ and $\ev{A_{pq,\sigma} A_{rs,\tau}}$.
We have classified all the possible terms appearing in the Hamiltonian and proposed a measurement clique for each type of the terms.
Especially, the measurement cliques for the terms $\ev{A_{pq,\sigma} A_{rs,\sigma}}$ with mutually-different $p,q,r,s$ have been constructed by finding the edge clique cover of the specific graph using the novel method based on the finite projective plane.
The total number of the distinct measurement cliques (or measurement circuits) is $2N^2 + O(N)$ with $2N$ being the number of spin orbitals (fermions), which exhibits better scaling compared with the previous method (BBO algorithm).
We have shown explicit quantum circuits to measure operators in the measurement cliques, and those circuits consist of $O(N^2)$ one- or two-qubit gates with depth $O(N)$ for Jordan-Wigner and parity mappings.
We have also performed numerical simulation for molecular Hamiltonians of hydrogen chains and evaluated the number of groups of simultaneously-measurable operators generated by our method as well as the number of measurement shots required to estimate the energy expectation values with sufficient accuracy. 
Evaluation of the expectation value of the Hamiltonian is one of the most fundamental subroutines among various algorithms for the applications of quantum computers to quantum chemistry, so our method can be utilized in broad studies on quantum algorithms.

As future work, it is intriguing to apply the proposed method to actual Hamiltonians in quantum chemistry and numerically investigate a standard deviation of the estimated energy expectation value with using our measurement cliques.
Another interesting direction is to generalize our method using the finite projective plane to higher order fermionic RDMs, such as $\ev{\Tilde{a}_p^\dag \Tilde{a}_q^\dag \Tilde{a}_r^\dag a_s a_t a_u}$.

\begin{acknowledgements}
KM is supported by JST PRESTO Grant No. JPMJPR2019.
This work is supported by MEXT Quantum Leap Flagship Program (MEXTQLEAP) Grant No. JPMXS0118067394 and JPMXS0120319794. We also acknowledge support from JST COI-NEXT program Grant No. JPMJPF2014.
This work initiated at QPARC Challenge 2022 (\url{https://github.com/QunaSys/QPARC-Challenge-2022}), an international hackathon for quantum information and chemistry, run by QunaSys Inc.
KM thanks Wataru Mizukami for fruitful discussions.
\end{acknowledgements}

\appendix

\section{Proof for several properties of points $\mathcal{S}$ on finite projective plane
\label{appsec: proofs for plane}}
In this Appendix, we prove several properties of points $\mathcal{S}$ on the finite projective plane described in Sec.~\ref{subsubsec: clique for same spin}.

\subsection{Three points in $\mathcal{S}$ never appear on a single line \label{appsubsec:no-three-point-proof}}
We prove the following:
\begin{lemma}
Three distinct points in $S$ never appear on a single line.
\end{lemma}
\begin{proof}
From the definition of $L_\alpha$, there is a single point $S(\Pi)=P_\alpha$ on $L_\alpha$. Similarly, $L_\beta(i)$ has only two points $S(\Pi)=P_\alpha$ and $S(i)=P_\gamma(i,i^2\mod \Pi)$ in $\mathcal{S}$.
For $L_\gamma(i,j)$, suppose that $L_\gamma(i,j)$ has three distinct points $S(a)$, $S(b)$ and $S(c)$ in $S$.
This implies
\begin{align}
 a^2 \equiv ia+j \mod \Pi, \:
 b^2 \equiv ib+j \mod \Pi.
\end{align}
Subtracting them results in
\begin{align}
 a^2-b^2 \equiv i(a-b) \mod \Pi. 
\end{align}
By using the assumption $a\neq b$, we reach
\begin{align}\label{eq:abi}
    a+b \equiv i \mod \Pi.
\end{align}
Similarly, we can show 
\begin{align}\label{eq:aci}
    a+c \equiv i \mod \Pi.
\end{align}
by using the assumption $a \neq c$.
These two equations~\eqref{eq:abi}\eqref{eq:aci} lead to $b=c$, but this contradicts with the assumption that three points $S(a), S(b)$ and $S(c)$ are distinct.
\end{proof}

\subsection{There is only one line that passes a point in $\mathcal{S}$
\label{appsubsec: exist-one-line}}
We prove the following property,
\begin{lemma}
 For any $S(k) \in \mathcal{S}$, there is only one line that passes $S(k)$ and does not pass any other points in $\mathcal{S}$.
\end{lemma}
\begin{proof}
For $S(\Pi) = P_\alpha$, $L_\alpha, L_\beta(0), \cdots, L_\beta(\Pi-1)$ are the only lines passing $S(\Pi)$ .
Among these lines, only the line $L_\alpha$ passes $S(\Pi)$ and does not pass any other points in $\mathcal{S}$.
For $S(k) = P_\gamma(k, k^2 \mod \Pi)$ with $k=0,\cdots,\Pi-1$, $L_\beta(k)$ and $L_\gamma(a, k^2 - ak \mod \Pi)$ for $a=0,\cdots,\Pi-1$ are the only lines passing $S(k)$.
The line $L_\beta(k)$ passes two points in $\mathcal{S}$, namely, $S(k)$ and $S(\Pi) = P_\alpha$.
For $l=0,\cdots,k-1,k+1,\cdots,\Pi-1$, the line $L_\gamma(k+l, k^2-(k+l)k \mod \Pi)$ passes the point $S(l) = P_\gamma(l, l^2 \mod \Pi)$ as well as $S(k)$.
We show that the remaining line $L_\gamma(2k, k^2 - 2k \cdot k \mod \Pi) = L_\gamma(2k, - k^2 \mod \Pi)$ do not pass the points in $\mathcal{S}$ other than $S(k)$ as follows.
Suppose that $L_\gamma(2k, - k^2 \mod \Pi)$ passes $S(m) = P_\gamma(m, m^2 \mod \Pi)$ for $S(m)\neq S(k)$.
Then we have
\begin{equation}
 m^2 \equiv 2km - k^2 \mod \Pi
 \: \Leftrightarrow \:
 (m-k)^2 \equiv 0 \mod \Pi.
\end{equation}
Since $\Pi$ is a prime, this means $m \equiv k \mod \Pi$.
This contradicts with $S(k) \neq S(m)$.
\end{proof}

\section{Measurement quantum circuits for Bravyi-Kitaev mapping
\label{appsec: meas. circ for PT-BK}}
Here, we provide further details on the Bravyi-Kitaev mapping, complementing the discussion in Sec.~\ref{subsubsec: diagonalization circuit}.
We follow the Fenwick-tree-based construction introduced in \cite{PhysRevA.95.032332}.

The Fenwick tree, also known as the binary indexed tree, is a data structure that allows us to efficiently perform range sums  of $N$-dimensional vector $\bm{x} = (x_0, x_1,...,x_{N-1})$.
Let us first represent an index $i$ of the vector as binary string: $i=\sum_{k=0}^{n} i_k 2^k$ where $n=\lfloor\log_2 N\rfloor$ and $i_k\in\{0,1\}$.
We define a least significant bit of an index $i$ as $\mathrm{LSB}(i) = i_{\tilde{k}} 2^{\tilde{k}}$ where $\tilde{k}$ is the smallest $k$ such that $i_k=1$.
Fenwick tree of $\bm{x}$ is a vector $\bm{y}=(y_0,y_1,...,y_{N-1})$ defined as,
\begin{align}\label{eq:fenwick-tree}
    y_i = \sum_{j=i+1-\mathrm{LSB}(i+1)}^{i} x_j.
\end{align}
Its structure can be illustrated as Fig.~\ref{fig:fenwick-tree}.
A range sum $\sum_{i=0}^{l-1} x_i$ for any $l>0$ can be obtained by 
summing up $O(\log(N))$ elements of $\bm{y}$.
Specifically, define a set of indices $P(l)$ by Algorithm \ref{alg:P}, which contains $O(\log(N))$ indices. Then, $\sum_{i=0}^{l-1} x_i = \sum_{i\in P(l)} y_i$. The Fenwick tree can also be updated in $O(\log(N))$ time. Specifically, consider adding a constant $c$ to $x_l$. This corresponds to adding $c$ to $y_i$ for every index $i$ in a set $U(l)$ defined by Algorithm \ref{alg:U}, which again contains $O(\log(N))$ indices.
\begin{figure}
    \centering
    \includegraphics{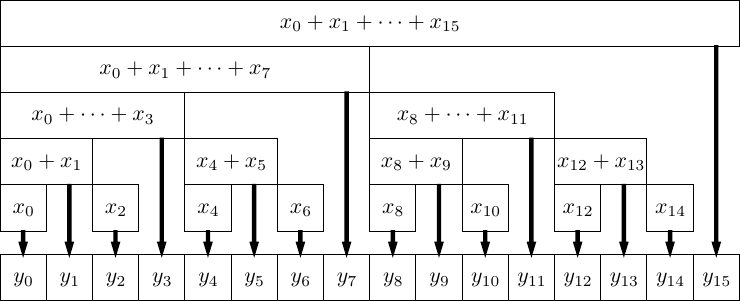}
    \caption{Illustration of Fenwick tree of size $n=16$ which is a basis of Bravyi-Kitaev mapping. An $n$-dimensional vector $\bm{x}$ can be transformed to the Fenwick tree $\bm{y}$ by taking appropriate range sums defined in Eq.~\eqref{eq:fenwick-tree}. Arrows represent what is stored as $y_i$.}
    \label{fig:fenwick-tree}
\end{figure}

\begin{figure}[t]
\begin{algorithm}[H]
    \caption{Generation of $P(l)$}\label{alg:P}
    \begin{algorithmic}
    \Require{$l$} 
    \State $i \gets l$
    \State $P \gets \{l-1\}$
    \While{$i>0$}
        \State $i \gets i - \mathrm{LSB}(i)$
        \State $P \gets P \cup \{i-1\}$
    \EndWhile
    \State \Return $P$
    \end{algorithmic}
\end{algorithm}
\end{figure}

\begin{figure}[t]
\begin{algorithm}[H]
    \caption{Generation of $U(l)$}\label{alg:U}
    \begin{algorithmic}
    \Require{$l$} 
    \State $i \gets l+1$
    \State $U \gets \{l\}$
    \While{$i<N$}
        \State $i \gets i + \mathrm{LSB}(i)$
        \State $U \gets U \cup \{i-1\}$
    \EndWhile
    \State \Return $U$
    \end{algorithmic}
\end{algorithm}
\end{figure}

The Bravyi-Kitaev mapping can be viewed as a fermion-to-qubit correspondence where we associate the fermionic state $(a^\dagger_0)^{x_0}(a^\dagger_1)^{x_1}\cdots (a_{N-1}^\dagger)^{x_{N-1}}\ket{\rm vac}$ to a qubit state $\ket{y_0y_1\dots y_{N-1}}$ using the Fenwick tree $\bm{y}$ of $\bm{x}$. Here, we assume the sum to construct $y_i$ is taken modulo 2.
Since Jordan-Wigner mapping associate $(a^\dagger_1)^{x_0}(a^\dagger_2)^{x_1}\cdots (a_{N-1}^\dagger)^{x_{N-1}}\ket{\rm vac}$ to a qubit state $\ket{x_0x_1\dots x_{N-1}}$, an fermionic operator $B_{\rm fermion}$ that is mapped to a qubit operator $B_{\rm JW}$ in Jordan-Wigner mapping corresponds to a qubit operator $B_{\rm BK}=WB_{\rm JW}W^\dagger$ in Bravyi-Kitaev mapping, where $W$ is a unitary such that $W\ket{x_0 x_1 \cdots x_{N-1}} = \ket{y_0 y_1 \cdots y_{N-1}}$.

Now, we consider how $W$ transforms single-qubit Pauli operators.
First, $WX_lW^\dagger = X_{U(l)}$ because applying $X_l$ to $\ket{x_0 x_1 \cdots x_{N-1}}$ corresponds to adding 1 to the $l$th bit and it is equivalent to adding 1 to every bit that is contained in $U(l)$ from the property of Fenwick tree.
Second, $WZ_lW^\dagger = Z_{P(l+1)}Z_{P(l)}$.
To see this, note that applying $Z_l$ to $\ket{x_0 x_1 \cdots x_{N-1}}$ corresponds to multiplying $(-1)^{x_l}$ to the state. On the other hand, applying $Z_{P(l)}$ to $\ket{y_0 y_1 \cdots y_{N-1}}$ corresponds to $(-1)^{\sum_{i\in P(l)} y_i} = (-1)^{\sum_{i=0}^{l-1} x_i}$.
$Z_{P(l+1)}Z_{P(l)}$, when applied to $\ket{y_0 y_1 \cdots y_{N-1}}$, therefore multiplies $(-1)^{\sum_{i=0}^{l} x_i}(-1)^{\sum_{i=0}^{l-1} x_i}=(-1)^{x_l}$ to the state, and this action is equivalent to $Z_l$ applied to $\ket{x_0 x_1 \cdots x_{N-1}}$.
Finally, $WY_lW^\dagger=iX_{U(l)}Z_{P(l+1)}Z_{P(l)}$ because $Y_l = iX_lZ_l $.
Note that both of $P(l+1)$ and $U(l)$ always contain $l$ and this allows us to write $WY_lW^\dagger= X_{U(l)\setminus \{l\}}Y_lZ_{P(l+1)\setminus\{ l\}}Z_{P(l)}$.

By observing these relations, it is straightforward to see that Eq. \eqref{eq: JW hopping term} holds.
First, $X_{2l}X_{2l+1}$ that appears in Jordan-Wigner mapping corresponds to $X_{U(2l)}X_{U(2l+1)}$.
From Algorithm \ref{alg:U}, $U(2l)\setminus U(2l+1)=\{2l\}$ and therefore $X_{U(2l)}X_{U(2l+1)} = X_{2l}$.
Second, consider $Y_{2l}Y_{2l+1}$ that appears in Jordan-Wigner mapping.
This corresponds to 
\begin{align}\label{eq:tmp}
    X_{U(2l)\setminus \{2l\}}Y_{2l}Z_{P(2l+1)\setminus\{2l\}}Z_{P(2l)} \cdot X_{U(2l+1)\setminus \{2l+1\}}Y_{2l+1}Z_{P(2l+2)\setminus\{2l+1\}}Z_{P(2l+1)}.
\end{align}
Note that $i<l$ holds if $i\in P(l)$ or $i \in P(l+1)\setminus \{l\}$ for any $l$.
Also, $i> l$ holds if $i\in U(l)\setminus \{l\}$.
This allows us to simplify Eq.~\eqref{eq:tmp} as
\begin{align}
    &X_{U(2l)\setminus \{2l\}}X_{U(2l+1)\setminus \{2l+1\}} \cdot Y_{2l}Y_{2l+1} \cdot Z_{P(2l+1)\setminus\{2l\}}Z_{P(2l)}Z_{P(2l+2)\setminus\{2l+1\}}Z_{P(2l+1)} \\
    &=X_{2l+1} \cdot Y_{2l}Y_{2l+1} \cdot Z_{2l}Z_{P(2l+2)\setminus\{2l+1\}}Z_{P(2l)} \\
    &= -Z_{2l+1}X_{2l}Z_{P(2l+2)\setminus\{2l+1\}}Z_{P(2l)}.
\end{align}
Note that we used $[U(2l)\setminus \{2l\}]\setminus[ U(2l+1)\setminus \{2l+1\}]=\{2l+1\}$, $P(2l+1) \setminus [P(2l+1)\setminus\{2l\}] = \{2l\}$ for the first equality, and the second equality follows from $X_{2l+1}Y_{2l+1}=iZ_{2l+1}$ and $Y_{2l}Z_{2l}=iX_{2l}$.

\section{Details of numerical calculation \label{appsec: numerics}}
We describe the details of the numerical calculation in Sec.~\ref{sec: numerics}.

\subsection{Details of QWC and GC grouping}\label{appsec:grouping-details}
Let us consider the qubit Hamiltonian $H_q = \sum_{i=1}^L c_i P_i$, where $P_i$ is $2N$-qubit Pauli operators and $c_i$ is its coefficient.
We assume that $P_i$ is not the identity operator.
In GC grouping with the ``sorted insertion" algorithm~\cite{crawford2021efficient}, the Pauli operators $\{ P_i \}_{i=1}^L$ are sorted by descending order of the absolute values of their coefficients $|c_i|$.
We denote the sorted Pauli operators $\{P'_i\}_{i=1}^L$.
We make the first group and include $P'_1$ in that group.
Then, we iterate the Pauli operators $P'_2, P'_3, \cdots$ and include it if it commutes with all terms in the group.
The creation of the first group ends when we reach $P'_L$, i.e., all of the remaining Pauli operators are checked.
The second group is made by including $P'_{k}$, where $k$ is the smallest number among the remaining Pauli operators that are not included in the first group.
We grow the second group by checking all of the remaining Pauli operators and picking them up if they commute with all terms in the group.
We repeat the creation of the groups in the same way until all Pauli operators are included in some group.

In QWC grouping with the sorted insertion algorithm, two $2N$-qubit Pauli operators $P = P^{(1)} \otimes \cdots \otimes P^{(2N)}$ and $Q = Q^{(1)} \otimes \cdots \otimes Q^{(2N)}$, where $P^{(i)}$ and $Q^{(i)}$ are single-qubit Pauli operators $I,X,Y,Z$ acting the site $i$, are said qubit-wise commuting if and only if $[P^{(i)}, Q^{(i)}] = 0$ holds for all $i=1,\cdots,2N$.
QWC grouping is performed in the same way as GC grouping, but we include a Pauli operator in some group when it satisfies qubit-wise commuting with all operators in the group.

\subsection{Calculation of the number of shots}\label{appsec: number of shots}
We omit the constant term in the qubit Hamiltonian and safely assume that $P_i$ is not the identity operator.
Suppose that we group the Hamiltonian as 
\begin{equation*}
 H_q = \sum_{g=1}^{G} O_g,  \:\: O_g = \sum_{j=1}^{L_g} c_{j}^{(g)} P_j^{(g)},
\end{equation*}
where $O_g$ is the $g$-th group, $c_j^{(g)} \: (P_j^{(g)})$ is the $j$-th coefficient (Pauli operator) of the $g$-th group ($j=1,\cdots,L_g)$, and $G$ is the total number of groups.
The Pauli operators in each group mutually commute so that we can simultaneously measure all operators in each group: $[P_j^{(g)}, P_{j'}^{(g)}] = 0$ for all $j, j'$.
When the total number of shots is $s$, the smallest standard deviation of the energy expectation value, $\ev{H_q}{\psi}$, for a quantum state $\ket{\psi}$ by the optimal shot allocation for the groups is known as~\cite{Rubin_2018,crawford2021efficient}
\begin{equation*}
 \sigma(s) = \frac{\gamma}{\sqrt{s}}, \:\: \gamma = \sqrt{ \sum_g \sqrt{ \sum_{j,j'} c_j^{(g)} c_{j'}^{(g)} \left(\ev{P_j^{(g)} P_{j'}^{(g)}}{\psi} - \ev{P_j^{(g)}}{\psi} \ev{P_{j'}^{(g)}}{\psi} \right) } }.
\end{equation*}
In the numerical calculations, we estimate the standard deviation of the energy by replacing the exportation values in this expression by Haar random average~\cite{gonthier2020identifying} ($\ev{P} = 0$ for any non-identity Pauli operator $P$):
\begin{equation*}
 \sigma'(s) = \frac{\gamma'}{\sqrt{s}}, \:\: \gamma' = \sqrt{ \sum_g \sqrt{ \sum_{j} \left(c_j^{(g)}\right)^2  } },
\end{equation*}
because the exact ground state of the Hamiltonian $\ket{\psi_\mr{GS}}$ is not available for the systems as large as $N \gtrsim 16$. 
The number of shots required to realize the standard deviation of $10^{-3}$ Hartree is calculated from this equation, i.e., $s_\mr{est} = (\gamma'/(10^{-3} \text{ Hartree}))^2$.

\bibliography{bib.bib}

\end{document}